%% file: arXiv.tex
\relax
\documentclass[11pt, letterpaper]{article}
\input{info-packages}

\title{Learning-Augmented Algorithms for Online TSP on the Line}
\author[1]{Themis Gouleakis}
\author[2]{Konstantinos Lakis}
\author[3]{Golnoosh Shahkarami}
\affil[1]{National University of Singapore,
\href{mailto:tgoule@nus.edu.sg}{tgoule@nus.edu.sg}}
\affil[2]{National and Kapodistrian University of Athens,
\href{mailto:konstlakis@gmail.com}{konstlakis@gmail.com}}
\affil[3]{Max Planck Institut für Informatik, Universität des Saarlandes, \href{mailto:gshahkar@mpi-inf.mpg.de}{gshahkar@mpi-inf.mpg.de}}

\date{}
\begin{document}

\maketitle
\begin{abstract}
    \noindent We study the online Traveling Salesman Problem (TSP) on the line augmented with machine-learned predictions. In the classical problem, there is a stream of requests released over time along the real line. The goal is to minimize the makespan of the algorithm. We distinguish between the \textit{open} variant and the \textit{closed} one, in which we additionally require the algorithm to return to the origin after serving all requests. The state of the art is a $1.64$-competitive algorithm and a $2.04$-competitive algorithm for the closed and open variants, respectively \cite{Bjelde:1.64}. In both cases, a tight lower bound is known \cite{Ausiello:1.75, Bjelde:1.64}.
    \medskip
    
    \noindent In both variants, our primary prediction model involves predicted \textit{positions} of the requests. 
    We introduce algorithms that (i) obtain a tight 1.5 competitive ratio for the closed variant and a 1.66 competitive ratio for the open variant in the case of perfect predictions, (ii) are robust against unbounded prediction error, and (iii) are smooth, i.e., their performance degrades gracefully as the prediction error increases.
    \medskip

    \noindent Moreover, we further investigate the learning-augmented setting in the \textit{open} variant by additionally considering a prediction for the last request served by the optimal offline algorithm. Our algorithm for this enhanced setting obtains a 1.33 competitive ratio with perfect predictions while also being smooth and robust, beating the lower bound of 1.44 we show for our original prediction setting for the open variant. Also, we provide a lower bound of 1.25 for this enhanced setting.
\end{abstract}

\section{Introduction}
\input{Introduction}
\section{Preliminaries}\label{section:Preliminaries}
\input{Preliminaries}
\section{Closed Variant}\label{section:closedLocations}
\input{Sections/ClosedLOCATIONS}

\section{Open Variant}\label{section:Open}
In this section, we consider the open variant. We have two 
prediction models for this variant. The first one is the $LOCATIONS$
prediction model and the second is the enhanced $LOCATIONS+FINAL$ model 
($LF$ in short). For both settings, we give algorithms and lower bounds.
\input{Sections/OpenLOCATIONS}

\input{Sections/OpenLF}

\section{Experimental Evaluation}
\input{Experiments}

\section{Conclusion}
\input{Conclusions}

\clearpage

\bibliographystyle{plain}
\bibliography{references}
\pagebreak

\appendix
\input{Appendix}

\end{document}

%% file: info-packages.tex
\usepackage{amsthm}
\usepackage[english]{babel}
\usepackage{titling}
\usepackage{a4wide}

\usepackage{cite}
\usepackage{tikz}
\usepackage{authblk}

\usepackage{booktabs}
\usepackage{thm-restate}

\usepackage[utf8]{inputenc} 
\usepackage[T1]{fontenc}    
\usepackage{hyperref}       
\usepackage{url}            
\usepackage{booktabs}       
\usepackage{amsfonts}       
\usepackage{nicefrac}       
\usepackage{microtype}      
\usepackage{xcolor}         
\usepackage{csquotes}

\usepackage{amsmath}
\usepackage{amssymb}
\usepackage{thmtools}
\usepackage{thm-restate}

\usepackage{graphics}
\usepackage{caption}
\usepackage{pgf}
\pdfpageattr {/Group << /S /Transparency /I true /CS /DeviceRGB>>}

\usepackage[ruled]{algorithm2e}

\usepackage{refcount}
\usepackage{fmtcount}
\usepackage{mleftright}

\makeatletter
\def\resetMathstrut@{%
  \setbox\z@\hbox{%
    \mathchardef\@tempa\mathcode`\[\relax
    \def\@tempb##1"##2##3{\the\textfont"##3\char"}%
    \expandafter\@tempb\meaning\@tempa \relax
  }%
  \ht\Mathstrutbox@\ht\z@ \dp\Mathstrutbox@\dp\z@}
\makeatother
\begingroup
  \catcode`(\active \xdef({\mleft\string(}
  \catcode`)\active \xdef){\mright\string)}
\endgroup
\mathcode`(="8000 \mathcode`)="8000

\newtheorem{definition}{Definition}
\newtheorem{lemma}{Lemma}

\newtheorem{claim}{Claim}[section]

\newenvironment{nscenter}
 {\parskip=0pt\par\nopagebreak\centering}
 {\par\noindent\ignorespacesafterend}

%% file: Introduction.tex
The Traveling Salesman Problem (TSP) is one of the most fundamental and widely studied problems in computer science, both in its offline version~\cite{offlineTSP}, where the input is known in advance, and the online version~\cite{Ausiello:1.75} where it arrives sequentially. In this paper, we consider the online Traveling Salesman Problem (TSP) on the real line. This version of the problem arises in real-world scenarios such as one dimensional delivery/collection tasks. Such tasks include the operation of elevator systems, robotic screwing/welding, parcel collection from massive storage facilities and cargo collection along shorelines \cite{Elevators, Shoreline}. Furthermore, one can think of other relevant practical settings such as the movement of emergency evacuation vehicles along a perilous highway, where there cannot be knowledge in advance regarding the time and location of persons requiring assistance. However, the great availability of data as well as the improved computer processing power and machine learning algorithms can make it possible for predictions to be made on these locations (e.g combining information from historical data, the weather forecast, etc). In a line of work that started a few years ago~\cite{lykouris2018competitive} and sparked a huge interest~\cite{NEURIPS2018_73a427ba,AntoniadisCE0S20,pmlr-v97-gollapudi19a,WangL20,DBLP:conf/innovations/0001DJKR20,Wei20,rohatgi2020near}, it has been demonstrated that such prior knowledge about the input of an online algorithm has the potential to achieve improved performance (i.e competitive ratio) compared to known algorithms (or even lower bounds) that do not use (resp. assume the absence of) any kind of prediction. Therefore, it is natural to consider ways to utilize this information in this problem using a so-called learning-augmented approach.

The input to our online algorithm consists of a set of requests, each associated with a position on the real line as well as a release time. An algorithm for this problem faces the task of controlling an agent that starts at the origin and can move with at most unit speed. The agent may serve a request at any time after it is released. The algorithm's objective is to minimize the makespan, which is the total time spent by the agent before serving all requests.
We have two different variants of the problem, depending on whether the agent is required 
to return to the origin after serving all the requests or not. This requirement exists 
in the \textit{closed} variant, while it does not in the \textit{open} variant. The makespan in the closed variant is the time it takes the agent 
to serve the requests \textit{and} return to the origin.


We quantify the performance of an online algorithm by its \textit{competitive ratio}, i.e., the maximum ratio of the algorithm's \textit{cost} to that of an optimal \textit{offline} algorithm $OPT$, over all possible inputs. We say that an algorithm with a competitive ratio of $c$ is $c$-competitive. Under this scope, the online TSP on the line has been extensively studied and there have been decisive results regarding lower and upper bounds on the competitive ratio for both variants of the problem. Namely, a tight bound of $\approx 1.64$ was given for the closed variant, while the corresponding value for the open 
variant was proven to be $\approx 2.04$ \cite{Ausiello:1.75, Bjelde:1.64}. 

\subsection{Our setup}\label{ourSetup}
First of all, to define our prediction model and algorithms, it is necessary to know the number of requests $n$
\footnote{Since this (slightly) modifies the original problem definition, the previous competitive ratio results for the classical problem do not necessarily hold for our setup even without predictions. We show in the Appendix that the bound of $1.64$ still holds for the closed variant and that a tight bound of $2$ holds for the open variant.}.
This setting shows up in various real world scenarios. For example, in the case of item collection from a horizontal/vertical 
storage facility, the capacity of the receiving vehicle, which awaits the successful collection of all items in order to deliver them to customers, dictates the number of items to be collected. We note that since $n$ is known, we can assume
that each prediction corresponds to a specific request determined by a given labeling, which is shared by both sets (requests and predictions).   
Under this assumption, we define the $LOCATIONS$ prediction model. In this model, the predictions are estimates for the positions of the requests. The error $\eta$ increases along with the maximum distance of a predicted location to the actual location of the identically labeled  request and is normalized by the length of the smallest interval containing the entire movement of the optimal algorithm.
We also define an enhanced prediction model for the open variant named $LOCATIONS+FINAL$ ($LF$ in short) that additionally specifies a request which is predicted to be served last by $OPT$. In this model, we additionally consider the error metric $\delta$, which increases with the distance of the predicted request to the
request actually served last by $OPT$. We also normalize $\delta$ in the same way as $\eta$. These models and their respective errors are defined formally in Section \ref{section:Preliminaries}.

\paragraph*{Properties of learning-augmented algorithms.} 
In the following we formalize the consistency, robustness and smoothness properties.   We say that an algorithm is:
    
    1) $\alpha$-\textit{consistent}, if it is $\alpha$-\textit{competitive} when there is no error.
    
    2) $\beta$-\textit{robust}, if it is $\beta$-\textit{competitive} regardless of prediction error.
    
    3) $\gamma$-\textit{smooth} for a continuous function $\gamma(err)$, if it is $\gamma(err)$-\textit{competitive}, where $err$ is the prediction error. Note that $err$ could potentially be a tuple of error types.

In general, if $c$ is the best competitive ratio achievable without predictions, it is desirable to have $\alpha < c$, $\beta \le k \cdot c$ for some constant $k$ and also the function $\gamma$ should increase from $\alpha$ to $\beta$ along with the error $err$. We note that $c, \alpha, \beta$ and the outputs of $\gamma$ may be functions of the input and not constant.

\subsection{Our contributions}\label{section:ourContributions}
Throughout this paper, we give upper and lower bounds for our three different settings (closed variant-$LOCATIONS$, open variant-$LOCATIONS$, and open variant-$LF$). The lower bounds refer to the case of perfect predictions and are established via different $attack$ strategies. That is, we describe the actions of an adversary $ADV$, who can control only the release times of the requests and has the goal of \textit{maximizing} the competitive ratio of any algorithm $ALG$. We emphasize that $ADV$ is given the power to observe $ALG$'s actions and act accordingly. In more detail, $ADV$ does not
need to specify the release times in advance, but can release a request at time $t$, taking the actions of $ALG$ until time $t$ into account. This is, in fact, the most powerful kind of adversary. The upper bounds are established via our algorithms and are defined for every value of the error(s). Recall that $\eta$ and $\delta$ refer to the two types of error we consider. Our algorithms and attack strategies are intuitively described in their respective sections. We now present the main ideas and our results.
\input{Contributions/ClosedLOCATIONS}
\input{Contributions/OpenLOCATIONS}
\input{Contributions/OpenLF}
We briefly summarize our results in Table \ref{table:results}. Note that the lower and upper bound entries correspond to the no error case. We emphasize that these results are for the case where the number of requests $n$ is known.
\input{TablesOfResults}
\subsection{Related work}
\input{RelatedWork}

%% file: Contributions/ClosedLOCATIONS.tex
\paragraph{Closed variant under \textit{LOCATIONS}.}
We will start by intuitively describing our algorithm for this setting and then continue with our lower bound. We design the algorithm $FARFIRST$. The main idea 
is that we first focus entirely on serving the requests on the side with the furthest extreme, switching to the other side when all such requests are served. When serving the requests on one side, we prioritize them by order of decreasing amplitude. The intuition is that we have the least possible amount of leftover work for our second departure from the origin, which limits the ways in which an adversary may attack us. We obtain the theorem below. More details are given in Section \ref{section:closedLocations}.
\begin{restatable}{thm}{FarFirstBound}
    \label{the:FarFirstBound}
    The algorithm $FARFIRST$ is $min\left\{ \frac{3(1 + \eta)}{2}, 3\right\}$-competitive.
\end{restatable}
We emphasize that for $\eta = 0$, this competitive ratio remarkably matches our lower bound of $1.5$, making $FARFIRST$ optimal.

Our lower bound for this setting 
is accomplished via an attack strategy that is analogous to a cunning magician's trick. Suppose that the magician keeps a coin inside one of their hands. They then ask a pedestrian to make a guess for which hand contains the coin. If the pedestrian succeeds, they get to keep the coin. However, the magician can always make it so that the pedestrian fails, for example by having a coin up each of their sleeves and producing the one not chosen by the pedestrian. One can draw an analogy from this trick to our attack strategy, which is described in Section \ref{section:closedLocations} in more detail. In this way, we obtain the theorem below.
\begin{restatable}{thm}{CompetitiveRatio}
    \label{the:CompetitiveRatio}
    For any $\epsilon > 0$, no algorithm can be $\left(1.5 - \epsilon\right)$-competitive for closed online TSP on the line under 
    the $LOCATIONS$ prediction model.
\end{restatable}

%% file: Contributions/OpenLOCATIONS.tex
\paragraph*{Open variant under \textit{LOCATIONS}.}

The algorithm we present for this setting is named
$NEARFIRST$. This algorithm first serves the requests on the side opposite to the one $FARFIRST$ would choose. Another divergence from $FARFIRST$ that should be noted is that for the side focused on second, $NEARFIRST$ prioritizes requests that are predicted to be \textit{closer} to the origin, since there is no requirement to return to it, thus avoiding unnecessary backtracking. More details 
about the algorithm and the proof of the following theorem are given in Section \ref{section:openLOCATIONS}.

\begin{restatable}{thm}{NearFirstBound}
    \label{the:NearFirstBound}
    The algorithm $NEARFIRST$ is $min\left\{f(\eta), 3\right\}$-competitive,
    where

    \begin{displaymath}
    f(\eta) = \left\{\begin{array}{lr}
        1 + \frac{2(1 + \eta)}{3 - 2\eta}, & \text{for }\quad \eta < \frac{2}{3}\\
        3, & \text{for }\quad \eta  \ge \frac{2}{3}
        \end{array}\right..
    \end{displaymath}

\end{restatable}
As in the previous setting, we utilize the "magician's trick" in order to design a similar attack strategy. We describe exactly how 
this is done in Section \ref{section:openLOCATIONS}.
This leads to the establishment of a lower bound, as stated below.

\begin{restatable}{thm}{CompetitiveRatioOpen}
    \label{the:CompetitiveRatioOpen}
    For any $\epsilon > 0$, no algorithm can be $\left(1.\overline{44} - \epsilon\right)$-competitive for open online TSP on the line under 
    the $LOCATIONS$ prediction model.
\end{restatable}

%% file: Contributions/OpenLF.tex
\paragraph*{Open variant under \textit{LOCATIONS+FINAL}.}

Our algorithmic approach to this setting is again similar 
to the one implemented in $NEARFIRST$. The difference is that
instead of choosing the side with the near extreme first, we choose the side 
whose extreme is further away from the predicted endpoint of $OPT$. We name this algorithm $PIVOT$, to emphasize that the prediction for the last request acts as a pivot for the algorithm to decide the first side it will serve. A theorem about $PIVOT$ is presented below, the proof of which has been given in Section \ref{section:openLF}.

\begin{restatable}{thm}{PivotBound}
    \label{the:PivotBound}
    The algorithm $PIVOT$ is $min\left\{f(\eta, \delta), 3\right\}$-competitive, where
    
    \begin{displaymath}
        f(\eta, \delta) = \left\{\begin{array}{lr}
        1 + \frac{1 + 2(\delta + 3\eta)}{3 - 2(\delta + 2\eta)}, & 3 - 2(\delta + 2\eta) > 0\\
        3, & 3 - 2(\delta + 2\eta) \le 0
        \end{array}\right..
    \end{displaymath}

\end{restatable}
For this setting, we reuse the attack strategy initially designed for the 
closed variant. The only difference is that we add another request at the origin 
with a release time of $4$. 
We explain how we derive the following theorem in Section \ref{section:openLF}.

\begin{restatable}{thm}{LFAttack}
    \label{the:LFAttack}
    For any $\epsilon > 0$, no algorithm can be $\left(1.25 - \epsilon\right)$-competitive for open online TSP on the line under 
    the $LF$ prediction model.
\end{restatable}

%% file: TablesOfResults.tex
\begin{table}
\begin{nscenter}
    \small
    \caption{Summary of results.\label{table:results}}
            \begin{tabular}{cccc} \toprule
                {Setting} & {Lower bound} & {Upper bound} & {Best competitive ratio}\\ \midrule
                {Closed variant without predictions}  & {1.64} & {1.64} & {1.64}\\ \midrule
                {Closed variant under LOCATIONS}  & {1.5}  & {1.5} & {$min\left\{\frac{3(1 + \eta)}{2}, 3\right\}$} \\ \bottomrule \toprule
                {Open variant without predictions}  & {2} & {2} & {2}\\ \midrule
                {Open variant under LOCATIONS}  & {1.44}  & {1.66} & {$min\left\{1 + \frac{2(1 + \eta)}{3 - 2\eta}, 3\right\}$} \\ \midrule
                {Open variant under LF}  & {1.25}  & {1.33} & {$min\left\{1 + \frac{1 + 2(\delta + 3\eta)}{3 - 2(\delta + 2\eta)}, 3\right\}$} \\ \bottomrule
            \end{tabular}
\end{nscenter}
\end{table}

%% file: RelatedWork.tex
\paragraph*{Online TSP.} The online TSP for a general class of metric spaces has been studied 
by Ausiello et al. in~\cite{Ausiello:1.75}, where the authors show lower bounds of $2$ for the open variant and $1.64$ for the closed variant. These bounds are actually shown on the real line. Additionally, a 2.5-competitive algorithm and a 2-competitive algorithm are given for the general open and closed variants respectively. A stronger lower bound of $2.04$ was shown for the open variant in~\cite{Bjelde:1.64} by Bjelde et al., where both bounds are also matched in the real line. For the restriction of the closed online TSP to the non-negative part of the real line, Blom et al.~\cite{Blom:1.5} give a tight $1.5$-competitive algorithm. By imposing a fairness restriction on the adversary, they also obtain a $1.28$-competitive algorithm. Jaillet and Wagner~\cite{Jaillet:advancedInfo} introduce the "online TSP with disclosure dates", where each request may also be communicated to the algorithm before it is released. The authors show improvements to the competitive ratios of previous algorithms as a function of the difference between disclosure and release dates. \paragraph*{Learning-augmented algorithms.} Learning-Augmented algorithms have received significant attention since the seminal work of Lykouris and Vassilvitskii~\cite{lykouris2018competitive}, where they introduced the online caching problem. Based on that model, Purohit et al.~\cite{NEURIPS2018_73a427ba} proposed algorithms for the ski-rental problem as well as non-clairvoyant scheduling. 
Subsequently, Gollapudi and Panigrahi~\cite{pmlr-v97-gollapudi19a}, Wang et al.~\cite{WangL20}, and Angelopoulos et al.~\cite{DBLP:conf/innovations/0001DJKR20} improved the initial ski-rental problem. The latter also proposed algorithms with predictions for the list update and bin packing problem and demonstrated how to show lower bounds for algorithms with predictions. Several works, including Rohatgi~\cite{rohatgi2020near}, Antoniadis et al.~\cite{AntoniadisCE0S20}, and Wei~\cite{Wei20}, improved the initial results regarding the caching problem.

The scheduling problems with machine-learned advice have been extensively studied in the literature. Lattanzi et al.~\cite{48659} considered the makespan minimization problem with restricted assignments, while Mitzenmacher~\cite{Mitzenmacher20} using predicted job processing times in different scheduling scenarios. Bamas et al.~\cite{BamasMRS20}, and Antoniadis et al.~\cite{DBLP:journals/corr/abs-2112-03082} focused on the online speed scaling problem using predictions for workloads and release times/deadlines, respectively.

There is literature on classical data structures. Examples include the indexing problem, Kraska et al.~\cite{DBLP:conf/sigmod/KraskaBCDP18}, bloom filters, Mitzenmacher~\cite{NEURIPS2018_0f49c89d}. 
Further learning-augmented approaches on online selection and matching problems~\cite{AntoniadisGKK20, DuttingLLV21} and a more general framework of online primal-dual algorithms~\cite{BamasMS20} also emerged, and there is a survey by Mitzenmacher et al.~\cite{MitzenmacherV20}.

\paragraph{Independent work.} 
Compared to the problem considered in this paper, a more general one, the online metric TSP, as well as a more restricted version in the half-line, have been studied in~\cite{independentWork} under a different setting, concurrently to our work. We note that only the closed variant is considered in~\cite{independentWork}. Since the prediction model is different (predictions for the positions as well as release times of the requests are given) and also a different error definition is used, the results are incomparable.

%% file: Preliminaries.tex
\paragraph*{The problem definition.}
In the online TSP on the line, an algorithm controls 
an agent that can move on the real line with at most unit speed. We have a set $Q = \{q_1, \ldots, q_n\}$ of $n$ requests. The algorithm receives the value $n$ as input. Each request $q$ has an associated position and release time. To simplify notation, whenever a numerical value is expected from a request $q$ (for a calculation, finding the minimum of a set, etc.) the term $q$ will refer to the \textit{position} of the request. Whenever we need the release time of a request, we will use $rel(q)$. Additionally, the algorithm receives as input a set $P = \{p_1, \ldots, p_n\}$ of predictions regarding the positions 
of the requests. That is, each $p_i$ attempts to approximate $q_i$. We assume without loss of generality that $Q$ always contains a request $q_0$ at the origin with release time $0$ and $P$ contains a perfect prediction $p_0 = 0$ for this request\footnote{This can be seen to be without loss of generality by 
considering a "handler" algorithm $ALG_0$ which adds this request/prediction pair to \textit{any} input and copies the actions of any of our algorithms $ALG$
for the modified input. We observe that $|OPT|$ is unchanged and $|ALG_0| = |ALG|$.}\label{originRequestFootnote}.


We use $t$ to quantify time. To describe the position 
of the agent of an algorithm $ALG$ at time $t \ge 0$, we use $pos_{ALG}(t)$.
We may omit this subscript when $ALG$ is clear from context. We can assume without loss of generality that $pos(0) = 0$. The speed limitation of the agent is given formally via $|pos(t') - pos(t)| \le |t' - t|, \: \forall t, t' \ge 0$. 
A request $q$ is considered served at time $t$ if $\exists \: t': \: pos(t') = q, \: rel(q) \le t' \le t$, i.e., the agent has moved to the request no earlier than it is released. We will say that 
a request $q$ is \textit{outstanding} at time $t$, if $ALG$ has not served 
it by time $t$, even if $rel(q) > t$, i.e. $q$ has not been released yet. Let $t_{serve}$ denote the first point in time when all requests have been served by the agent. Also, let $|ALG|$ denote the makespan of an algorithm $ALG$, for either of the two variants. Then, for the open variant $|ALG| = t_{serve}$ while for the closed one $|ALG| = min\{t: pos(t) = 0, \: t \ge t_{serve}\}$. For any sensible algorithm, this is equivalent to $t_{serve} + |pos(t_{serve})|$, since the algorithm knows the number of requests and will immediately return to the origin after serving the last one. The objective is to minimize the 
value $|ALG|$, utilizing the predictions.

\input{Notation}

\input{Prediction Models/LOCATIONS}

\input{Prediction Models/LF}

%% file: Notation.tex
\paragraph*{Notation.}
We define $L = min(Q)$ and $R = max(Q)$. Recall that $Q$ contains a request at the origin and thus $L \le 0$ and $R \ge 0$. We refer to each of these requests as an \textit{extreme} request. 
If $|L| > |R|$, we define $Far = L, Near = R$. Otherwise, $Far = R, Near = L$.
That is, $Far$ is the request with the largest distance from the origin out 
of all requests. Then, $Near$ is simply the other extreme. We will also refer to the value $|q|$ as the \textit{amplitude} of request $q$.

We denote with $O(t)$ the set of outstanding requests at time $t$. Then, 
$L_O(t) = min(O(t) \cup \{pos(t)\})$ and $R_O(t) = max(O(t) \cup \{pos(t)\})$.
We additionally define $L(t) = min(O(t) \cup \{R\})$ and $R(t) = max(O(t) \cup \{L\})$. The difference between $L_O(t), R_O(t)$ and $L(t), R(t)$ is that the former also consider the position of $ALG$ to determine the interval that must be traveled to serve all the requests while the latter assume that $ALG$ is already somewhere inside the interval of outstanding requests (which may not be true due to $ALG$ moving to a bad prediction).

For technical reasons, we have two different notions (and thus terms) for unreleased requests. We will use the same notation for convenience but the terms we introduce will slightly differ for the closed and open variant. For the closed variant, we let $L_{lim} = 0, R_{lim} = 0$ while $L_{lim} = R, R_{lim} = L$ for the open variant. Thus, we define
\begin{displaymath}
L_U[t] = min(\{q \in Q \: : \: rel(q) \ge t\} \cup \{L_{lim}\}),
\end{displaymath}
\begin{displaymath}
R_U[t] = max(\{q \in Q \: : \: rel(q) \ge t\} \cup \{R_{lim}\}),
\end{displaymath}
while also letting 
\begin{displaymath}
L_U(t) = min(\{q \in Q \: : \: rel(q) > t\} \cup \{L_{lim}\}),
\end{displaymath}
\begin{displaymath}
R_U(t) = max(\{q \in Q \: : \: rel(q) > t\} \cup \{R_{lim}\}).
\end{displaymath}
The former will be used to prove the upper bounds while the latter will be used to prove the lower bounds.

Recall that we assume without loss of generality that $0 \in P$. We define $L_P = min(P)$ and $R_P = max(P)$.
We will say that a prediction $p$ is (un)released/outstanding/served if the associated request $q$ is (un)released/outstanding/served. For a request $q \in Q$ matched with a prediction $p \in P$, we define $\pi(q) = p$ and $\pi^{-1}(p) = q$.
That is, the function $\pi$ takes us from the requests to the associated predictions and $\pi^{-1}$
takes us from the predictions to the requests.

%% file: Prediction Models/LOCATIONS.tex
\paragraph*{The \textit{LOCATIONS} prediction model.}
We now introduce the $LOCATIONS$ prediction model. Let $q_1, ..., q_n$ be a
labeling of the requests in $Q$. The predictions consist of the values $p_1, ..., p_n$,
where each $p_i$ attempts to predict the position of $q_i$.

\paragraph*{Error definition for the \textit{LOCATIONS} prediction model.}
To give an intuition for the metric we will introduce, let us first describe what it means for a
prediction to be bad. In any well-posed definition, the further $p_i$ is from $q_i$, the worse it should be graded.
However, we must also take into account the "scale" of the problem, meaning the length of the interval $[L, R]$ that must be traveled by any algorithm, including $OPT$. The larger this interval, the more lenient our penalty 
for $p_i$ should be. Therefore, we define the error as
\begin{displaymath}
\eta [Q, P] = \frac{max_i\{|q_i - p_i|\}}{|L| + |R|}.
\end{displaymath}
Additionally, we define $M = \eta \cdot (|L| + |R|)$.
\paragraph*{An important lemma for the \textit{LOCATIONS} prediction model.} 
We now present a lemma that gives us some intuition about this prediction model.
\begin{restatable}{lemma}{LOCATIONS}
    \label{lem:LOCATIONS}
    Let $L_P = min(P), R_P = max(P)$.
    Then, $|L_P| \ge |R_P|$ implies $|L| \ge |R| - 2M$, and $|R_P| \ge |L_P|$ implies $|R| \ge |L| - 2M$.
\end{restatable}
\input{Proofs/Lemma1}

%% file: Proofs/Lemma1.tex
    \begin{proof}
    The following claim constitutes the main part of our proof.
    \begin{claim}\label{claim:extremesM}
        $|L_P - L| \le M$ and $|R_P - R| \le M$.
    \end{claim}
    \begin{proof}
        
    If $L_P = \pi(L) \implies |L_P - L| = |\pi(L) - L|$
    or $L = \pi^{-1}(L_P) \implies|L_P - L| = |L_P - \pi^{-1}(L_P)|$ then we see that $|L_P - L| \le M$. Thus, we assume 
    the contrary for the rest of the proof.

    Since $L_P$ is by definition the leftmost prediction, we know that $L_P < \pi(L)$. 
    Additionally, since $L$ is the leftmost request, we know that $L < \pi^{-1}(L_P)$.

    Let $X \le Y \le Z \le W$ represent the values of the set $\{L, \pi^{-1}(L_P), L_P, \pi(L)\}$ in 
    ascending order. It should be easy to see that $X$ must be equal to either 
    $L$ or $L_P$. Otherwise, one of $L_P < \pi(L)$ or $L < \pi^{-1}(L_P)$ is violated, leading to 
    a contradiction. We distinguish two cases.

    \textbf{Case 1.} $X = L$. In this case, $L_P$ comes after $L$ but before $\pi(L)$
    in the $X, Y, Z, W$ ordering. Therefore, $|L - L_P| \le |L - \pi(L)| \le M$.

    \textbf{Case 2.} $X = L_P$. Similarly, $L$ comes after $L_P$ but before $\pi^{-1}(L_P)$
    in the $X, Y, Z, W$ ordering. Thus, $|L_P - L| \le |L_P - \pi^{-1}(L_P)| \le M$.
    
    The inequality $|R_P - R| \le M$ can be seen in a symmetric way.
    \end{proof}
    Using this claim, we can now conclude the proof of Lemma \ref{lem:LOCATIONS}.
     We focus on the case $|L_P| \ge |R_P|$; the other case is symmetrical. By Claim \ref{claim:extremesM}, we have
        \begin{displaymath}
        |R_P - R| \le M \implies |R_P| - |R| \le M \implies |R| \le |R_P| + M.
    \end{displaymath}
    Additionally, we have
    \begin{displaymath}
        |L_P - L| \le M \implies |L_P| - |L| \le M \implies |L| \ge |L_P| - M.
    \end{displaymath}
    Combining these inequalities with $|L_P| \ge |R_P|$ proves the lemma.
    \end{proof}

%% file: Prediction Models/LF.tex
\paragraph*{Enhanced prediction model for the open variant.} Motivated by the performance of our algorithm under the $LOCATIONS$ prediction model, we enhance it with a prediction $f'$ which attempts to guess the label $f$ of a request on which $OPT$ may finish. We name this new model $LF$ (short for $LOCATIONS+FINAL$). The error $\eta$ is unchanged. We also introduce a new error metric $\delta$. 
Let $q_{f'}$ be the request
associated with the prediction $p_{f'}$. We then choose $q_f$ to be a request on which $OPT$ may finish that minimizes the distance to $q_{f'}$. We then define the new error as
\begin{displaymath}
\delta[Q, q_f, q_{f'}] = \frac{|q_{f'} - q_f|}{|L| + |R|}.
\end{displaymath}
Similarly to before, we define $\Delta = \delta \cdot (|L| + |R|)$.

%% file: Sections/ClosedLOCATIONS.tex
In this section, we consider the closed variant under the $LOCATIONS$ 
prediction model. We provide the $FARFIRST$ algorithm, which obtains a 
competitive ratio of $1.5$ with perfect predictions and is also smooth 
and robust. Additionally, we give an attack strategy that implies a lower 
bound of $1.5$ for the competitive ratio of any algorithm in this setting, 
making $FARFIRST$ optimal.
\paragraph*{The $FARFIRST$ algorithm.}
Before giving the algorithm, we define the $FARFIRST$
ordering on the predictions of an input. For simplicity, 
we assume that the furthest prediction from the origin is positive. Let 
$r_1, \ldots, r_a$ be the positive predictions in descending 
order of amplitude and $l_1, \ldots, l_b$ be the negative 
predictions ordered in the same way. The $FARFIRST$ ordering is $r_1, \ldots, r_a, l_1, \ldots, l_b$. Any predictions on the origin are placed in the end. Ties are broken via an arbitrary label ordering.

We present the algorithm through 
an update function used whenever a request 
is released. This update function returns 
the plan of moves to be executed until the 
next release of a request. Note that $ext(side, set)$ returns 
the extreme element of the input set in the side specified, where $side = true$ means the right side.
Also, the $\oplus$ symbol is used to join moves one after another. When all the moves are executed, the agent waits for the next release. This only happens when waiting on a prediction.
\begin{nscenter}
\begin{algorithm}
    \caption{$FARFIRST$ update function.}\label{alg:FarFirstSmooth}
    \SetKwInOut{Input}{Input}
        \SetKwInOut{Output}{Output}
    
        \Input{Current position $pos$, set $O$ of unserved released requests, first unreleased prediction $p$ in $FARFIRST$ ordering or 0 if none exist, the side $farSide$ with the furthest prediction from the origin.
               }
        
        \Output{A series of (unit speed) moves to carry out until the next request is released.}
    $posSide \gets (pos > 0)$\;
    $pSide \gets (p > 0)$\;

    \lIf{$pos = 0$}{
        $posSide \gets farSide$
    }
    
    \lIf{$p = 0$}{
        $pSide \gets \overline{posSide}$
    }
    
    \Return $move(ext(posSide, O \cup \{pos\})) \oplus move(ext(pSide, O \cup \{p\})) \oplus move(p)$\;
\end{algorithm}
\end{nscenter}
In order to give some further intuition on $FARFIRST$, we first give the definition of a \textit{phase}. 
\begin{definition}\label{phase}
A phase of an algorithm $ALG$ is a time interval $[t_s, t_e]$ such that $pos_{ALG}(t_s) = 0$, $pos_{ALG}(t_e) = 0$ and $pos_{ALG}(t') \neq 0, \: \forall \: t' \in (t_s, t_e)$. That is, $ALG$ starts and ends a phase at the origin and does not cross the origin at any other time during the phase.
\end{definition}
In the following, when we refer to the \textit{far} side, we mean the side with the furthest \textit{prediction} from the origin. The \textit{near} side is the one opposite to that. We see that $FARFIRST$ works in at most three phases. The first phase ends when all predictions on the far side have been released
and the agent has managed to return to the origin with no released and outstanding request on the far side. During this phase, any request on the far side is served as long as $FARFIRST$ does not move closer to the origin than the far side's extreme unreleased prediction. Note that some \textit{surprise} requests may appear, i.e., far side requests that were predicted to lie on the near side. These requests are also served in this phase. The second phase lasts while at least one prediction is unreleased.
During this phase, the agent serves any request released on the near side, using the predictions as guidance, similarly to the first phase. Requests released on the far side are ignored during this phase. Note that no surprises can occur here, since all far side predictions were released during the first phase. A third phase may exist if some requests were released on the far side during the second phase. These requests' amplitudes are bounded by $M$, since they were predicted to be positioned on the near side. This simple algorithm is consistent, smooth and robust, as
implied by the following theorem.
\FarFirstBound*
Let us begin with the intuition behind the proof. The 3-robustness is seen using an absolute worst case scenario in which $FARFIRST$ is $|OPT|$ units away from the 
origin at time $|OPT|$ (due to the unit speed limitation), and all the requests to serve are on the opposite side. For the consistency and smoothness, we note that $|OPT| \ge 2(|Near| + |Far|)$. It is therefore sufficient to prove that
\begin{displaymath}
|FARFIRST| - |OPT| \le |Near| + |Far| + 3\eta \cdot (|Near| + |Far|) = |Near| + |Far| + 3M.
\end{displaymath}
We refer to the left hand side as the \textit{delay} of $FARFIRST$. We now see why this bound holds intuitively. 
We first describe a worst case scenario. In this scenario, $OPT$ first serves the near side completely, and then does the same for the far side, without stopping.  Let $t_e$ denote the end time of the first phase. We see that $t_e \le |OPT| + M$, because $FARFIRST$ follows the fastest possible route serving the requests on the far side, except for a possible delay of $M$ attributable to a misleading prediction. Note that in this worst case, all requests on the near side must have been released by $t_e$. Therefore, $FARFIRST$ accumulates an 
extra delay of at most 2 times the maximum amplitude of these requests. By Lemma \ref{lem:LOCATIONS}, this value is at most $|Near| + |Far| + 2M$. There are 
also other possibilities than this worst case, but they also can incur a delay of at most $|Near| + |Far| + 3M$, because $|OPT|$ and $|FARFIRST|$ both increase when such cases occur.

We now give the formal proof of Theorem \ref{the:FarFirstBound}. 
\input{Proofs/FARFIRST}

\paragraph{A $1.5$-attack.}
Based on the magician analogy presented in Section \ref{section:ourContributions}, we design an attack strategy 
that yields the following theorem.

\CompetitiveRatio*

We first describe in high level the main ideas in the 
proof of this theorem. In our attack strategy, we have arbitrarily many requests evenly placed 
in the interval $[-1, 1]$. The more of these requests we have,
the closer the competitive ratio we achieve will be to $1.5$.

For a given set $Q_X$ of request \textit{positions}, we 
now describe how the release times of the requests at these positions
are decided. The strategy is that we release requests on both sides 
as long as $ALG$ has not yet approached a released request. This is the first
phase of releases and it is structured in such a way that $OPT$ 
could begin serving either of the two sides as fast as possible. In the 
magician analogy we described, this corresponds to the time before the pedestrian 
chooses a hand.

When $ALG$ approaches a released request, we "freeze" the requests on 
$ALG$'s side. That is, if $ALG$ moves close to a released request on the left side, 
say one placed at $-\frac{1}{2}$, all requests in the interval $[-\frac{1}{2}, 0]$
have their release time delayed such that $OPT$ can still serve the entire right 
side and then come back to serve the left side by $t = 4$. This corresponds 
to the magician producing the coin on the right hand while the pedestrian has 
chosen the left hand. However, $ALG$ 
is now faced with a dilemma. Should it wait for these "frozen" requests
or should it travel all the way to $1$ in order to serve the right side first?
We will see that both options are bad, in the sense that $|ALG|$ can be 
seen to be arbitrarily close to $6$. We now proceed with the formal proof.

\input{Proofs/ClosedAttack}

%% file: Proofs/FARFIRST.tex
We will first prove the robustness part of this theorem.

\begin{lemma}\label{FarFirstRobustness}
The algorithm $FARFIRST$ is 3-robust.
\end{lemma}

\begin{proof}

Let $t_{f}$ denote the latest release time for a fixed instance of the problem.
We assume w.l.o.g. that $pos_{FARFIRST}(t_{f}) \le 0$. Note that after $t_f$, $FARFIRST$ will move to $L(t_f)$, then to $R(t_f)$ and then back to the origin. Thus, we observe that
\begin{equation}\label{FarFirstCleanupTime}
    |FARFIRST| = t_f + |pos(t_f) - L(t_f)| + |L(t_f) - R(t_f)|  + |R(t_f)|.
\end{equation}
We distinguish two 
cases based on the position of $FARFIRST$ at time $t_f$.
\textbf{Case 1.} $pos(t_{f}) \ge L(t_{f})$. In this case, we see that
\begin{displaymath}
{\eqref{FarFirstCleanupTime} \implies
|FARFIRST| = t_f + pos(t_f) - L(t_f) + R(t_f) - L(t_f) + |R(t_f)| \le}
\end{displaymath}
\begin{displaymath}
{t_f + 2(|L(t)| + |R(t_f)|) \le t_f + 2(|L| + |R|) \le 2|OPT| \le 3|OPT|.}
\end{displaymath}
\textbf{Case 2.}  $pos(t_{f}) < L(t_{f})$. Similarly, we have
\begin{displaymath}
{\eqref{FarFirstCleanupTime} \implies
|FARFIRST| = t_f + L(t_f) - pos(t_{f}) + R(t_f) - L(t_f) + |R(t_f)| \le}
\end{displaymath}
\begin{displaymath}
    {2t_f + 2|R(t)| \le 2t_f + 2|R| \le 3|OPT|.}
\end{displaymath}
\end{proof}
Now, to prove Theorem \ref{the:FarFirstBound}, it remains to show the consistency/smoothness part, which is given by the following lemma.

\begin{lemma}\label{FarFirstSmoothness}
   The algorithm $FARFIRST$ is $f(\eta)$-smooth, where $f(\eta) = \frac{3(1 + \eta)}{2}$.
\end{lemma}
To prove this lemma, we will bound $FARFIRST$'s \textit{delay}, i.e. the value $|FARFIRST| - |OPT|$, as shown below.
\begin{equation}\label{FarFirstDelay}
    |FARFIRST| - |OPT| \le |Near| + |Far| + 3M
\end{equation}
This is sufficient because Equation \eqref{FarFirstDelay} along with the elementary bound of $|OPT| \ge 2\left(|Near| + |Far|\right)$ prove Lemma \ref{FarFirstSmoothness}. Thus, we now state and prove the following 
claim.

\begin{claim}\label{claim:FarFirstDelay}
    For any input, we have $|FARFIRST| - |OPT| \le |Near| + |Far| + 3M$.
\end{claim}

We assume w.l.o.g. that $|L_P| \le |R_P|$. Thus, by Lemma \ref{lem:LOCATIONS} we see that
\begin{displaymath}
|L| \le |R| + 2M \implies 2|L| \le |R| + |L| + 2M =
\end{displaymath}
\begin{displaymath}
|Near| + |Far| + 2M \implies 2|L| + M \le |Near| + |Far| + 3M.
\end{displaymath}
Therefore, it also suffices to show that 
\begin{equation}\label{FarFirstLDelay}
|FARFIRST| - |OPT| \le 2|L| + M
\end{equation}
We now describe the way in which we will prove Equation \eqref{FarFirstLDelay}
or Equation \eqref{FarFirstDelay}. Recall the definition of a \textit{phase} given in Definition \ref{phase}.
We note here that we will also use the term \textit{delay} 
to refer to how much later a phase ends compared to $|OPT|$. We will use another claim stating that for a single phase, $FARFIRST$ will serve the requests on the side of the phase as fast as possible or $OPT$ is seen to finish at most $M$ time units before $FARFIRST$ finishes the phase, thus "resetting" the delay counter. Using this claim for the (at most) three phases of $FARFIRST$, we can indeed show Claim \ref{claim:FarFirstDelay}.
In the following, we will consider a phase in the right side of the origin. We now define a term that is similar to $R_U[t]$. 

Let
${R_U}'[t] = max(\{q:\:q\in Q, \: rel(q) \ge t, \: \pi(q) > 0\} \cup \{0\})$. It should be obvious that ${R_U}'[t] \le R_U[t]$. Note that when ${R_U}'[t] = 0$, this means that all requests associated with positive predictions have been released by time $t$, thus prompting $FARFIRST$ to conclude the phase.

We should explain here that $R_U[t]$ works as a \textit{block} for $OPT$ (since it has to wait for a request to be released in order to serve it). Similarly ${R_U}'[t]$ works in the same way for $FARFIRST$, which must serve all requests associated with 
a positive prediction before ending the phase. We observe a useful relationship between these two blocks, which implies that if $FARFIRST$ is blocked on a request to the right of $M$, then so is $OPT$. This relationship is encapsulated in the following claim.

\begin{claim}\label{claim:blocks}
    If $R_U[t] > M$, then ${R_U}'[t] = R_U[t]$.
\end{claim}

\begin{proof}
    We see that $\pi(R_U[t]) \ge R_U[t] - M > 0$. Therefore, $\pi(R_U[t])$ is a \textit{positive} prediction and thus ${R_U}'[t] = R_U[t]$.
\end{proof}

The next definition is about the time it would take (after $t$) for $FARFIRST$ to serve all requests (to the right of $R_U[t]$) and then reach 
$R_U[t]$. If this is not more than $M$, we can see that $FARFIRST$ is not too far behind $OPT$. If it is more 
than $M$, we shall see that $FARFIRST$ has enough information to progress through the phase as fast as possible.

    $D(t)$ denotes the least amount of time 
    necessary to serve all requests to the right of $R_U[t]$ (assuming they have been released) and then move to $R_U[t]$,
    starting at position $pos_{FARFIRST}(t)$. This amounts 
    to
        \begin{displaymath}
            D(t) = |pos_{FARFIRST}(t) - R_O(t)| + |R_O(t) - R_U[t]|.
        \end{displaymath}
This function exhibits a useful bound property. If it 
drops to $M$ or below at some time $t$, it can only increase 
above $M$ again due to a request release. This property is 
described more formally in the following claim. But first, 
another useful definition is given.

We define $R_P[t]$ as the rightmost positive prediction released at time $t$ or later. If no such prediction exists, then $R_P[t] = 0$. Note that $FARFIRST$ never moves to the left of this prediction. We now give a relevant claim.
\begin{claim}\label{claim:pToRU}
    $|R_P[t] - R_U[t]| \le M$. 
\end{claim}
\begin{proof}
We can see that $|R_P[t] - {R_U}'[t]| \le M$ by the definition of these terms. If $R_U[t] > M$, the claim immediately follows by Claim \ref{claim:blocks}. 

Otherwise, $R_U[t] \le M$. We have $R_P[t] \ge 0 \implies R_U[t] - R_P[t] \le M$. Additionally, we know that ${R_U}'[t] \le R_U[t]$ and $R_P[t] \le {R_U}'[t] + M \implies R_P[t] - R_U[t] \le M$, concluding the proof.
\end{proof}
\begin{claim}\label{claim:chainSaw}
    Let $t_{drop}$ be a time point such that $D(t_{drop}) \le M$. 
    If $t_{next}$ is the earliest release time of a request 
    after $t_{drop}$, then
        \begin{displaymath}
            D(t') \le M, \: \forall \: t' \in [t_{drop}, t_{next}].
        \end{displaymath}
\end{claim}
\begin{proof}
    Observe that $R_U[t]$ is constant throughout the interval $[t_{drop}, t_{next}]$.
    Let $R_U$ denote this constant value. The same is true for $R_P[t]$,
    which is always equal to a specific value $p$. We split the interval 
    $[t_{drop}, t_{next}]$ into three parts. 

    \textbf{Part 1.} This part lasts while $FARFIRST$ is moving 
    towards a released request to the right of $max\{p, pos(t)\}$. This decreases 
    the value $|pos(t) - R_O(t)|$ while $|R_O(t) - R_U|$ is constant and
    thus $D(t)$ cannot increase.

    \textbf{Part 2.} This part lasts while $FARFIRST$ is moving towards 
    $p$. No released requests exist to the right of $pos(t)$ during 
    this time, since that is taken care of in Part 1. Thus, we 
    have $R_O(t) = max(pos(t), R_U)$. Either way, we see that $D(t) = |R_U - pos(t)|$
    during this part. At the start of this part, we have $D(t) \le M$. When 
    $p$ is reached, we still have $D(t) \le M$, because $p$ has a distance 
    of at most $M$ to $R_U$ by Claim \ref{claim:pToRU}. Thus, we have $D(t) \le M$ 
    throughout this part also. 

    \textbf{Part 3.} This part lasts while $pos(t) = p$, i.e. $FARFIRST$ 
    is waiting on top of $p$. It can be seen that $D(t)$ is constant 
    throughout this part and also not larger than $M$.
\end{proof}

We are now ready to present and prove the main claim we discussed.

\begin{claim}\label{claim:zoomOrCatchup}
    Assume without loss of generality that $FARFIRST$ is focusing on the right side during a phase. $FARFIRST$ finishes this phase as fast as possible or does so at most $M$ time units after OPT finishes.
    More 
    precisely, if the phase spans the time 
    interval $[t_s, t_e]$ and the rightmost request served during this phase is $R_{phase}$, then
        \begin{displaymath}
            (t_e - t_s = 2|R_{phase}|) \vee  (t_e - |OPT| \le M).
        \end{displaymath}
\end{claim}
    First of all, note that if at least one request is unreleased at time $t_e$, then obviously $|OPT| \ge t_e \implies t_e - |OPT| \le M$. Thus, we can assume in the following that \textit{all} requests will have been released before the end of the phase.
    
    We now draw our attention to a point in time that is very 
    central to our proof.
 
Let $t_{release}$ be the latest release time of a positive request associated with a positive prediction. Note 
        that ${R_U}'[t] = 0, \: \forall \: t > t_{release}$. Then, we define
            \begin{displaymath}
                t_{chase} = min\{t : \: t_s \le t \le t_{release}, \: (D(t') > M, \: \forall \: t < t' \le t_{release})\}.
            \end{displaymath}
    Intuitively, $t_{chase}$ signifies the start of a series of unit speed moves executed by $FARFIRST$ that lead to a \textit{final state} in which $FARFIRST$ has made sufficient progress through the phase and is also not too far behind $OPT$. After it reaches this state, it is easier to prove Claim \ref{claim:zoomOrCatchup}. We now describe what exactly we mean by this state.
    \begin{definition}[Final state]
        We say that $FARFIRST$ has reached a final state in a phase at time $t_{state}$ if $pos(t_{state}) \le M$ and there are no outstanding requests or unreleased predictions to the right of position $M$.
    \end{definition}
    We see why this final state is important in the following claim.
    \begin{claim}\label{claim:ProlongingOrFast}
        If $FARFIRST$ is in a final state at time $t_{state}$ with $pos(t_{state}) = x_{state} \le M$, then
            \begin{displaymath}
                (t_e - t_{state} = |x_{state}|) \vee  (t_e - |OPT| \le M).
            \end{displaymath}
        \begin{proof}
            If $FARFIRST$ moves straight to the origin after $t_{state}$, the first part is true. 
            On the other hand, there are only two possible ways for $FARFIRST$ \textit{not} to return straight to the origin, both of which provide new lower bounds for $|OPT|$, thus "resetting" the delay. One of them is for $FARFIRST$ to wait for a prediction $p \le M$ with $\pi^{-1}(p) \le 0$. Because $OPT$ also has to wait for this request and since $FARFIRST$ will ignore it for this phase, the delay is seen to be at most $M$ after such a case. The other case is for a request $q$ on the right side to be released that was predicted to be on the left side, implying that $q \le M$. Again, it can take up to $2|q|$ time units for $FARFIRST$ to serve this request and return but also $OPT$ needs to spend at least $|q|$ time units to terminate after it is released. Again, the delay is seen to be at most $M$.
    \end{proof}
\end{claim}

Now that our goal has been somewhat clarified, we proceed with the main part of the proof. We now show that after $t_{chase}$, $FARFIRST$ moves to a final state as soon as possible.
\begin{claim}\label{claim:FastFinalState}
    Let $x_{state} = min\{M, R_O(t_{chase})\}$. Then, $FARFIRST$ reaches a final state at time $t_{state}$ and $pos(t_{state}) = x_{state}$,
    where
        \begin{displaymath}
            t_{state} = t_{chase} + |pos(t_{chase}) - R_O(t_{chase})| + |R_O(t_{chase}) - x_{state}|.
        \end{displaymath}
\end{claim}
\begin{proof}
    This can be seen by considering the moves followed by $FARFIRST$ after $t_{chase}$. First of all, we show that $FARFIRST$ moves straight to $R_O(t_{chase})$, starting at $t_{chase}$. If $pos(t_{chase}) = R_O(t_{chase})$, the claim is obvious. Thus, by the definition of $R_O(t)$, we can assume that $pos(t_{chase}) < R_O(t_{chase})$. It suffices to show that $FARFIRST$ moves to the right until it reaches $R_O(t_{chase})$. We split this move into two possible parts.
    
    \textbf{Part 1.} This part only applies if $pos(t_{chase}) < R_O(t_{chase}) - M$. In this part, 
    we show that $FARFIRST$ moves straight to the point $R_O(t_{chase}) - M$. Indeed, if $R_O(t_{chase})$ is released at some point during this part, then $FARFIRST$ will surely move to $R_O(t_{chase})$ (let alone $R_O(t_{chase}) - M$) in order to serve it. If $R_O(t_{chase})$ is not released during this part, then $R_U[t] = R_O(t_{chase})$ throughout this part. But because $R_U[t] = R_O(t_{chase}) > M$, Claim \ref{claim:blocks} implies that 
    $R_U[t] = {R_U}'[t] \implies R_P[t] \ge {R_U}'[t] - M = R_O(t_{chase}) - M$. Therefore, $FARFIRST$ will move to $R_O(t_{chase}) - M$ because of the predictions in this case.
    
    \textbf{Part 2.} This part refers to the move from
     the point $x = max\{R_O(t_{chase}) - M, pos(t_{chase})\}$
      to $R_O(t_{chase})$. In any case (whether Part 1
      applies or not), when $pos(t_x) = x$, $R_O(t_{chase})$ 
      \textit{must} have been released. Indeed, assume for the  
      sake of contradiction that $rel(R_O(t_{chase})) > t_x$.
      Note then that $R_U[t] = R_O(t_{chase})$ until $rel(R_O(t_{chase}))$.
      We can see that 
            \begin{displaymath}
                |pos(t) - R_O(t_{chase})| = |pos(t) - R_U[t]|  \le M \: 
                \forall \: t \in [t_x, rel(R_O(t_{chase}))].
            \end{displaymath}
      If $R_U[t] \le M$ for such $t$, the claim can be seen by 
      noting that $R_U[t] - M \le 0$ and that $R_U[t] + M \ge {R_U}'[t] + M$ and 
      because $FARFIRST$ won't exit the interval $[0, {R_U}'[t] + M]$ due to 
      $R_P[t]$.
      
      Otherwise,
      by Claim \ref{claim:blocks}, we have that $R_U[t] = {R_U}'[t]$ for such $t$.
      This means that  
      $FARFIRST$ will not move to the left of $R_O(t_{chase}) - M$
      due to $R_P[t]$ and also the rightmost point that may be travelled 
      to is $R_O(t_{chase}) + M$, again because of $R_P[t]$. But that 
      would mean that there exists $t' \: : \: t_{chase} < t' \le t_{release}$ with 
      $D(t') \le M$, a contradiction. Therefore since $R_O(t_{chase})$
      is released at $t_x$, $FARFIRST$ will move towards it immediately.

    It now remains to show that $FARFIRST$ will move 
    to $x_{state}$ immediately after reaching $R_O(t_{chase})$. If $x_{state} = R_O(t_{chase})$, the
    claim is obvious. Therefore, we can assume that $x_{state} = M$ and $x_{state} < R_O(t_{chase})$. We again split this move into 
    two parts.
    
    \textbf{Part 1.} This part lasts while $t \le t_{release}$ and $x_{state} = M$ has not yet been reached. This means that throughout this part, we have
    \begin{displaymath}
    D(t) > M \implies R_U[t] < pos(t) - M \implies {R_U}'[t] + M < pos(t) \implies R_P[t] < pos(t).
    \end{displaymath}
Thus, since $R_P[t]$ and $R_U[t]$ are always to the left of $pos(t)$, $FARFIRST$ neither stops to wait for a prediction nor backtracks to serve a request during this part.
    
\textbf{Part 2.} This part starts after Part 1 and lasts until $x_{state} = M$ is reached. Again, $FARFIRST$ trivially does not stop to wait for a prediction, since all the positive ones are released by now. Additionally, we can see that $R_U[t] \le M$ for this part, since all unreleased predictions are not positive. 
Thus, $R_U[t] \le pos(t)$ also holds for this part, prohibiting backtracking. 

We can see that in both parts $FARFIRST$ moves to the left with unit speed.

Therefore, two unit speed moves are followed after $t_{chase}$, one to $R_O(t_{chase})$ and one to $x_{state} = min\{M, R_O(t_{chase})\}$.
    Also, after these moves, $FARFIRST$ has reached a final state, because no outstanding request or unreleased prediction exists to the right of $x_{state} \le M$. End of proof.
\end{proof}

We will now use claims \ref{claim:ProlongingOrFast} and \ref{claim:FastFinalState} along with the definition of $t_{chase}$ to prove Claim \ref{claim:zoomOrCatchup}. 

\begin{proof}[Proof of Claim \ref{claim:zoomOrCatchup}.]
We distinguish two cases.

\textbf{Case 1.} $t_{chase} = t_{s}$. In this case, Claim \ref{claim:FastFinalState} implies that $FARFIRST$ reaches a final state by time $t_{state} = t_s + |R_{phase}| + |R_{phase} - x_{state}|$, where $x_{state} = min\{R_{phase}, M\}$. Then, by Claim \ref{claim:ProlongingOrFast}, we have 
    \begin{displaymath}
        (t_e - t_{state} = |x_{state}|) \vee  (t_e - |OPT| \le M) \implies
        (t_e - t_s = 2|R_{phase}|) \vee  (t_e - |OPT| \le M).
    \end{displaymath}
Thus, Claim \ref{claim:zoomOrCatchup} holds in this case.

\textbf{Case 2.} $t_{chase} > t_{s}$.
In this case, we first show that $t_{state} \le |OPT| + max\{M -  R_O(t_{chase}), 0\}$, where 
$t_{state}$ is as described in Claim \ref{claim:FastFinalState}.
To achieve this, we note that $D(t_{chase}) = M$. Indeed, 
let $t_{prev}$ be the latest release time before $t_{chase}$,
or $t_s$ if none exist. If $D(t') > M$ for all $t' \: : \: t_{prev} < t' \le t_{chase}$,
then the definition of $t_{chase}$ is violated. Thus, there exists a 
$t' \: : \: t_{prev} < t' \le t_{chase}$ such that $D(t') \le M$ and we have $D(t_{chase}) \le M$
by Claim \ref{claim:chainSaw}. We now distinguish two 
subcases.

\textbf{Case 2.1.} $M \le R_O(t_{chase}) \implies x_{state} = M$. In this case, we see that
    \begin{displaymath}
        D(t_{chase}) \le M \implies |pos(t_{chase}) - R_O(t_{chase})| + |R_O(t_{chase}) - R_U[t_{chase}]| \le M \implies
    \end{displaymath}
    \begin{displaymath}
        R_U[t_{chase}] \ge |pos(t_{chase}) - R_O(t_{chase})| + |R_O(t_{chase}) - M| \implies
    \end{displaymath}
    \begin{displaymath}
        t_{chase} + R_U[t_{chase}] \ge t_{chase} + |pos(t_{chase}) - R_O(t_{chase})| + |R_O(t_{chase}) - M| \implies
    \end{displaymath}
    \begin{displaymath}
        |OPT| \ge t_{state}.
    \end{displaymath}
\textbf{Case 2.2.} $M > R_O(t_{chase}) \implies x_{state} = R_O(t_{chase})$. Because $D(t_{chase}) \le M$, we 
must have
    \begin{displaymath}
        |pos(t_{chase}) - R_O(t_{chase})| \le M - R_O(t_{chase}) + R_U[t_{chase}] \implies    
    \end{displaymath}
    \begin{displaymath}
        t_{state} \le |OPT| + M - R_O(t_{chase}).  
    \end{displaymath}
In both of these subcases, by Claim \ref{claim:ProlongingOrFast} we have
    \begin{displaymath}
        (t_e - t_{state} = |x_{state}|) \vee  (t_e - |OPT| \le M) \implies
        (t_e = |x_{state}| + t_{state}) \vee  (t_e - |OPT| \le M) \implies
    \end{displaymath}
    \begin{displaymath}
        (t_e \le |OPT| + M) \vee  (t_e - |OPT| \le M).
    \end{displaymath}
Claim \ref{claim:zoomOrCatchup} is now proved in all cases.
\end{proof} 
We can finally use Claim \ref{claim:zoomOrCatchup} to prove Claim \ref{claim:FarFirstDelay} by showing that Equation \eqref{FarFirstLDelay} or Equation \eqref{FarFirstDelay} holds.
\begin{proof}[Proof of Claim \ref{claim:FarFirstDelay}.]
Let $t_s(1), t_e(1)$ be the start and end times of the first phase of $FARFIRST$ and $t_s(2), t_e(2)$ are similarly defined for the second phase. We can see that $t_s(1) = 0$ and $t_s(2) = t_e(1)$. We distinguish two cases based on the possible existence of a third phase. 

\textbf{Case 1.} No requests are released on the right side after $t_e(1)$. Thus, we see that $|FARFIRST| = t_e(2)$. By Claim \ref{claim:zoomOrCatchup}, we see that $t_e(1) \le |OPT| + M$. Using Claim \ref{claim:zoomOrCatchup} for the second phase also, we see that
\begin{displaymath}
    (t_e(2) - t_s(2) = 2|L|) \vee  (t_e(2) - |OPT| \le M) \implies
\end{displaymath}
\begin{displaymath}
    (t_e(2) - t_e(1) = 2|L|) \vee  (t_e(2) - |OPT| \le M) \implies \eqref{FarFirstLDelay} \implies \eqref{FarFirstDelay}.
\end{displaymath}
\textbf{Case 2.} At least one request is released 
on the right side after $t_e(1)$. Let $q_{M}$ be the rightmost such request. We can see that $q_M \le M$, since it is necessarily associated with a non-positive prediction. We can see that $|FARFIRST| = t_e(2) + 2|q_M|$. We also know that $|OPT| \ge rel(q_M) + |q_M| \ge t_e(1) + |q_M|$. By Claim \ref{claim:zoomOrCatchup} for the second phase, we
have
    \begin{displaymath}
        (t_e(2) - t_s(2) = 2|L|) \vee  (t_e(2) - |OPT| \le M) \implies
    \end{displaymath}
\begin{displaymath}
    (t_e(2) - t_e(1) = 2|L|) \vee  (t_e(2) - |OPT| \le M) \implies
\end{displaymath}
\begin{displaymath}
    (t_e(2) + 2|q_M|  = 2|L| + t_e(1) + 2|q_M|) \vee  (t_e(2) + 2|q_M| \le |OPT| + M + 2|q_M|) \implies
\end{displaymath}
\begin{displaymath}
    (|FARFIRST| \le |OPT| + 2|L| + |q_M|) \vee  (|FARFIRST| \le  |OPT| + M + 2|q_M|) \implies
\end{displaymath}
\begin{displaymath}
    (|FARFIRST| \le |OPT| + 2|L| + M) \vee  (|FARFIRST| \le  |OPT| + 3M) \implies
\end{displaymath}
\begin{displaymath}
    \eqref{FarFirstLDelay} \vee \eqref{FarFirstDelay} \implies \eqref{FarFirstDelay}.
\end{displaymath}
\end{proof}
We can now use Claim \ref{claim:FarFirstDelay} to prove Lemma \ref{FarFirstSmoothness}.
\begin{proof}[Proof of Lemma \ref{FarFirstSmoothness}.]
We know that $|OPT| \ge 2(|L| + |R|)$, since it must at least travel 
to both $L$ and $R$ and back. Also, by Claim \ref{claim:FarFirstDelay}, we have
$|FARFIRST| - |OPT| \le |Near| + |Far| + 3M = |L| + |R| + 3M$. These inequalities imply
\begin{displaymath}
    \frac{|FARFIRST|}{|OPT|} =  1 + \frac{|FARFIRST| - |OPT|}{|OPT|} \le 1 + \frac{|L| + |R| + 3M}{2(|L| + |R|)} = \frac{3(1 + \eta)}{2}.
\end{displaymath}
\end{proof}
We can finally prove Theorem \ref{the:FarFirstBound}.
\begin{proof}[Proof of Theorem \ref{the:FarFirstBound}.]
By Lemma $\ref{FarFirstRobustness}$, $FARFIRST$ is 3-robust. Additionally,
by Lemma \ref{FarFirstSmoothness}, $FARFIRST$ is $(\frac{3(1 + \eta)}{2})$-smooth. Thus, the theorem holds.
\end{proof}

%% file: Proofs/ClosedAttack.tex
We describe a family $F_C$ of inputs that is structured as follows. For a given rank $n \ge 2$, we place exactly $n$ requests evenly spaced across the interval $[-1, 1]$.
For an instance $f$ of the family $F_C$, $\alpha(f)$ is 
defined as the distance between any consecutive pair of requests 
in $f$.

\begin{claim}\label{claim:distanceOfRequests}
    If an instance $f$ of the family $F_C$ has rank $n$, then $\alpha(f) = \frac{2}{n - 1}$.
\end{claim}
\begin{proof}
    There are $n$ requests that delimit an interval of 
    length $2$. Thus, there are $n - 1$ equal subintervals, whose lengths' sum is 
    equal to $2$. Therefore, each subinterval has length $\frac{2}{n - 1}$.
\end{proof}
All that is left to determine is the release times of the requests. We split the release times into two "phases". The first phase takes place for as long as $L_U(t) < pos_{ALG}(t) < R_U(t)$. During this phase, $ADV$ releases any request with distance $d$ from the origin at time $2 - d$. Note that this release method allows $OPT$ to eagerly start serving any side of the origin first without waiting for requests to release.

Now for the second phase's releases, assuming
that $ALG$ exits the interval to its left side (i.e. commits to the left extreme),
the requests to the right side are released as during the first phase. However, for any unreleased request to the left side with distance $d$ from the origin, its release time is delayed to $4 - d$.
If $ALG$ exits from the right instead, $ADV$ releases the left requests 
as in the first phase and delays the right requests. The input (positions and release
times of requests) is now fully specified.
For the following, we will define $t_{commit}$ as the start time of the second phase. That is,
\begin{displaymath}
    t_{commit} = min(\{t \: : \: \lnot (L_U(t) < pos_{ALG}(t) < R_U(t))\}).
\end{displaymath}
We immediately observe the following inequality, which guarantees 
that the second phase of the requests starts in a timely manner.
\begin{claim}\label{claim:commit_1_2}
$1 \le t_{commit} \le 2$.
\end{claim}
\begin{proof}
    For the sake of contradiction, assume that $t_{commit} < 1$.
    Then, $|pos(t_{commit})| \ge 1$ since $[L_U(t), R_U(t)] = [-1, 1]$ for $t < 1$.
    However, because any algorithm is limited to unit speed, $t_{commit} < 1 \implies |pos(t_{commit})| < 1$,
    a contradiction. 
    
    On the other hand, assume that $t_{commit} > 2$. This means 
    that $L_U(2) < pos(2) < R_U(2)$. But, since the first phase has not 
    stopped until $t = 2$, we have $L_U(2) \ge 0, R_U(2) \le 0$, which clearly leads 
    to a contradiction.
\end{proof}
We now state a lemma ensuring that $OPT$ finishes in the 
absolute least time possible for any such input. This 
allows us to maximize the competitive ratio we 
achieve against $ALG$.
\begin{lemma}\label{lem:OPT4}
For any instance $f$ in the family $F_C$, $|OPT| = 4$.
\end{lemma}
\begin{proof}
    We observe that the requests of one side (the one 
    $ALG$ did not exit from) are released such that $OPT$
    can serve them all and return to the origin by $t = 2$.
    Additionally, the other side's requests are released 
    such that $OPT$ never has to stop for them either, i.e. 
    it can serve them all and return to the origin by $t = 4$. 
    Thus, $|OPT| = 4$.
\end{proof}
However, $ALG$ has commited to one side (by exiting the interval) and we will prove that it requires at least $6 - 2\alpha(f)$ time units to terminate. This will be our main lemma.
We state it here for reference but will prove it later.
\begin{restatable}{lemma}{ALGSlow}
    \label{lem:ALGSlow}
    For any instance $f$ in the family $F_C$ and for any $ALG$, we have that $|ALG| \ge 6 - 2\alpha(f)$.
\end{restatable}
Before proving this lemma, we give some more claims. For the following, we assume without loss of generality that $ALG$ exits 
the interval $[L_U(t), R_U(t)]$ from the left, i.e. it commits to the left side.
\begin{claim}\label{claim:posCommitInterval}
    For any instance $f$ in the family $F_C$,
        \begin{displaymath}
            L_U(t_{commit}) - \alpha(f) \le pos_{ALG}(t_{commit}) \le L_U(t_{commit}).
        \end{displaymath}
\end{claim}
\begin{proof}
The claim can be seen by examining two cases. If $ALG$ exited the unreleased requests interval itself by moving out of it, then $pos_{ALG}(t_{commit}) = L_U(t_{commit})$. In the other case, $ALG$ was forced out of the interval by a request release. Thus, right before this release (which occurs at precisely $t_{commit}$), $ALG$ was inside the interval. The previous interval was $[L_U(t_{commit}) - \alpha(f), R_U(t_{commit}) + \alpha(f)]$. Thus, the inequality holds.
\end{proof}
We now draw attention to one particular value, which 
constitutes the backbone of our attack. We define $d_{commit} = |L_U(t_{commit})|$, where $t_{commit}$ is the 
start of the second phase of releases.
    
We now show some claims that allow us to use this value to get a lower bound for $|ALG|$.
\begin{claim}\label{claim:timeCommitInterval}
    For any instance $f$ in the family $F_C$, $t_{commit} \ge 2 - d_{commit} - \alpha(f)$.
\end{claim}
\begin{proof}
    If $L_U(t_{commit}) = -1 \implies d_{commit} = 1$, then by Claim \ref{claim:commit_1_2}, we have 
    $t_{commit} \ge 1 \ge 2 - d_{commit} - \alpha(f)$. Therefore, we can assume that 
    a request $L_{prev}$ exists with $L_{prev} = L_U(t_{commit}) - \alpha(f) < L_U(t_{commit})$. We see that if $t < 2 - d_{commit} - \alpha(f)$,
    then $L_U(t) \le L_{prev}$, because $L_{prev}$ is unreleased until $2 - d_{commit} - \alpha(f)$. Therefore, 
    $t_{commit} \ge 2 - d_{commit} - \alpha(f)$, since otherwise we would have $L_U(t_{commit}) \le L_{prev}$, a contradiction.
\end{proof}
The following claim states that $ALG$ has essentially made no progress 
until $t_{commit}$. If $t_{commit}$ is close to 2, we can easily see 
why this is bad for $ALG$. On the other hand, an early commit means 
that $d_{commit}$ will be large (due to Claim \ref{claim:timeCommitInterval}), posing problems again for $ALG$.
\begin{claim}\label{claim:noServesFirstPhase}
$ALG$ has not served any request during the first phase, i.e. up to time $t_{commit}$.
\end{claim}
\begin{proof}
This is due to the fact that $ALG$ has not exited the interval of unreleased requests until $t_{commit}$. Therefore, it cannot have moved to a released request. Since $ALG$ has to move to a request to serve it, the claim holds.
\end{proof}
Now we are ready to prove Lemma \ref{lem:ALGSlow}.
\begin{proof}[Proof of Lemma \ref{lem:ALGSlow}.]
Let us examine the options that $ALG$ has in order  to terminate after $t_{commit}$. By Claim \ref{claim:noServesFirstPhase}, we know that $ALG$ has not yet served the requests at $-1, L_U(t_{commit}), 1$. We examine cases based on the order in which it chooses to do so from $t_{commit}$ on.\\

\textbf{Case 1.} $ALG$ serves $1$ before $-1$. Then, $ALG$ at the very least needs to travel from $pos_{ALG}(t_{commit})$ to $1$, then to $-1$ and then back to the origin. Using Claims \ref{claim:timeCommitInterval} and \ref{claim:posCommitInterval}, this takes at least
    \begin{displaymath}
        |ALG| \ge t_{commit} + |pos_{ALG}(t_{commit}) - 1| + |1 - (-1)| + |-1 - 0| \ge 
    \end{displaymath}
    \begin{displaymath}
        (2 - d_{commit} - \alpha(f)) + (d_{commit} + 1) + 2 + 1 = 6 - \alpha(f).
    \end{displaymath}
\textbf{Case 2.} $ALG$ serves $-1$, then $1$ and then $L_U(t_{commit})$. Again using Claims \ref{claim:timeCommitInterval} and \ref{claim:posCommitInterval}, this takes
    \begin{displaymath}
        |ALG| \ge t_{commit} + |pos_{ALG}(t_{commit}) - (-1)| + |(-1) - 1| + |1 - L_U(t_{commit})| + |L_U(t_{commit}) - 0| \ge     \end{displaymath}
    \begin{displaymath}
        (2 - d_{commit} - \alpha(f)) + (1 - d_{commit} - \alpha(f)) + 2 + (1 + d_{commit}) + d_{commit} = 6 - 2\alpha(f).
    \end{displaymath}
\textbf{Case 3.}  $ALG$ serves $L_U(t_{commit})$ before $1$. In this case, $ALG$ has to first wait for $L_U(t_{commit})$ to be released and then go to serve $1$. Because $L_U(t_{commit})$ is a request to the left of the origin released during the second phase, we have
    \begin{displaymath}
        |ALG| \ge rel(L_U(t_{commit})) + |L_U(t_{commit}) - 1| + |1 - 0| \ge (4 - d_{commit}) + (1 + d_{commit}) + 1 = 6.
    \end{displaymath}
These cases are exhaustive and thus Lemma \ref{lem:ALGSlow} is proved.
\end{proof}
With all the above, we can finally prove Theorem \ref{the:CompetitiveRatio}.
\begin{proof}[Proof of Theorem \ref{the:CompetitiveRatio}.]
By Lemma \ref{lem:ALGSlow} and Lemma \ref{lem:OPT4}, we have a competitive ratio of at least $\frac{6 - 2\alpha(f)}{4}$ for any algorithm $ALG$. By Claim \ref{claim:distanceOfRequests}, we can see that $\alpha(f)$ can be arbitrarily small and thus this competitive ratio can be arbitrarily close to $1.5$, proving our claim.
\end{proof}

%% file: Sections/OpenLOCATIONS.tex
\subsection{The \textit{LOCATIONS} prediction model}\label{section:openLOCATIONS}
Under the $LOCATIONS$ prediction model, we design the 
$NEARFIRST$ algorithm, which achieves a competitive ratio of $1.\overline{66}$
with perfect predictions and is also smooth and robust. We complement this 
result with a lower bound of $1.\overline{44}$ using a similar attack 
strategy to the one used for the closed variant.
\paragraph*{The $NEARFIRST$ algorithm.}
As we mentioned in the introduction, $NEARFIRST$ is similar to $FARFIRST$ 
and actually slightly simpler. In essence, $NEARFIRST$ simply picks a direction 
in which it will serve the requests. Then, it just serves the requests either 
from left to right or from right to left, using the predictions as guidance. 
The pseudocode for $NEARFIRST$ is given below. Recall that $move(x) \oplus move(y)$
is used to indicate a move to $x$ followed by a move to $y$.
\begin{nscenter}
\begin{algorithm}[!htb]
    \caption{$NEARFIRST$ update function.}\label{alg:NearFirst}
    \SetKwInOut{Input}{Input}
        \SetKwInOut{Output}{Output}
    
        \Input{Current position $pos$, set $O$ of unserved released requests, set $P$ of predictions.
               }
        
        \Output{A series of (unit speed) moves to carry out until the next request is released.}
    $P' \gets$ the unreleased predictions in $P$\;
    \If{$P'$ is empty}{
        \lIf{$pos < \frac{max(O) + min(O)}{2}$}{
            \Return $move( min(O)) \oplus move(max(O))$
            }
            \lElse{
            \Return $move(max(O)) \oplus move( min(O))$
            }
    }
    \lIf{$|min(P)| < |max(P)|$}{
        \Return $move(min(P' \cup O)) \oplus move(min(P'))$
    }
    \lElse{
        \Return $move(max(P' \cup O)) \oplus move(max(P'))$
    }

\end{algorithm}
\end{nscenter}
We present the following theorem regarding the competitive ratio of $NEARFIRST$.
\NearFirstBound*
We first describe the main ideas used in the proof of this theorem. As in the case of $FARFIRST$,
the 3-robustness holds because at time $|OPT|$, $NEARFIRST$ has "leftover work" of at most 
$2|OPT|$ time units (to return to the origin and then copy $OPT$). For the consistency/smoothness,
we draw our attention to the request $q_f$ served last
by $OPT$. For the following, we assume that $NEARFIRST$ serves the requests left to right. Let $d = |q_f - R|$. We will show that the delay of $NEARFIRST$ is bounded by $M + d$. Let $t_{q_f}$ be the time when $NEARFIRST$ has served all requests to the left of $q_f$, including $q_f$. It turns out that $t_{q_f} \le |OPT| + M$, because $NEARFIRST$ serves this subset of requests as fast as possible, except for a possible delay of $M$ due to a misleading prediction. Then, in this worst case, $NEARFIRST$ accumulates an extra delay of at most $d$, proving our claim.

Finally, we bound $OPT$ from below as a function of $d$. We see that $OPT$ can either serve the requests ${L, R, q_f}$ in the order $L, R, q_f$ or in the order $R, L, q_f$. The worst case is the latter, where we see that $|OPT| \ge 2|R| + |L| + (|L| + |R| - d) = 3|R| + 2|L| - d$. Since $d \le |L| + |R|$, we obtain
    \begin{displaymath}
    \frac{|NEARFIRST|}{|OPT|} = 1 + \frac{|NEARFIRST| - |OPT|}{|OPT|} \le 1 + \frac{M + |L| + |R|}{2|R| + |L|}.
    \end{displaymath}
Because $NEARFIRST$ considers $L$ the near extreme due to the predictions, by Lemma \ref{lem:LOCATIONS} we find that $|R| \ge \frac{1 - 2\eta}{2}(|L| + |R|)$, which in turn proves our bound.

We now give the formal proof of Theorem \ref{the:NearFirstBound}.
First of all, we present two lemmas that are very important. Their proofs
are deferred to the Appendix, since they also refer to the $PIVOT$ 
algorithm, which is introduced later.

\begin{restatable}{lemma}{OpenRobustness}\label{lem:OpenRobustness}
    Let $ALG$ be either $NEARFIRST$ or $PIVOT$.
    Then, $ALG$ is 3-robust.
    \end{restatable}
\begin{restatable}{lemma}{OpenDelay}\label{lem:OpenDelay}
    Let $ALG$ be either $NEARFIRST$ or $PIVOT$.
    Also, let $q_f$ be the request served last 
    by $OPT$. Assume without loss of generality that $ALG$ serves 
    requests from left to right. Let $d = |q_f - R|$.
    Then, we have $|ALG| - |OPT| \le M + d$.
\end{restatable}
\input{Proofs/NEARFIRST}

\paragraph*{A $1.\overline{44}$-attack.}
For this setting, we use an attack strategy that is very similar to 
the one introduced in Section \ref{section:closedLocations}. This allows us
to obtain the following theorem.
\CompetitiveRatioOpen*
\input{Proofs/OpenAttackLOCATIONS}

%% file: Proofs/NEARFIRST.tex
The robustness part of Theorem \ref{the:NearFirstBound} is implied by Lemma \ref{lem:OpenRobustness}.

Now, to prove Theorem \ref{the:NearFirstBound}, it remains to show the consistency/smoothness part, which is given by the following lemma.

\begin{lemma}\label{lem:NearFirstSmoothness}
    The algorithm $NEARFIRST$ is $f(\eta)$-smooth for $\eta < \frac{2}{3}$, where $f(\eta) = 1 + \frac{2(1 + \eta)}{3 - 2\eta}$.
\end{lemma}
\begin{proof}
    Assume w.l.o.g. that $NEARFIRST$ serves the left extreme first. By Lemma \ref{lem:OpenDelay}, we see that $|NEARFIRST| - |OPT| \le M + d$,
    where $d$ is the distance of $R$ to the request 
    $q_f$ served last by $OPT$.
    We distinguish two cases based on the order in which $OPT$
    serves the requests in $\{L, R, q_f\}$.
    \textbf{Case 1.} $OPT$ serves $L, R$ and then $q_f$. It can be 
    seen then that
        \begin{displaymath}
            \frac{|NEARFIRST|}{|OPT|} \le 1 + \frac{M + d}{2|L| + |R| + d}.
        \end{displaymath}
    The derivative of the right hand side with respect to $d$ is
        \begin{displaymath}
            \dfrac{|R| + 2|L| - M}{\left(d+|R|+2|L|\right)^2}.
        \end{displaymath}
    Because $\eta < \frac{2}{3}$, it must hold that $M < \frac{2}{3}(|L| + |R|) \implies \dfrac{|R| + 2|L| - M}{\left(d+|R|+2|L|\right)^2} > 0$.
    Thus, we may maximize $d$ to get an upper bound that is valid for any value of $d$.
    Because $d \le |L| + |R|$ we see that
        \begin{displaymath}
            \frac{|NEARFIRST|}{|OPT|} \le 1 + \frac{M + |L| + |R|}{3|L| + 2|R|}.
        \end{displaymath}
    \textbf{Case 2.} $OPT$ serves $R, L$ and then $q_f$. It can be 
    seen then that
        \begin{displaymath}
            \frac{|NEARFIRST|}{|OPT|} \le 1 + \frac{M + d}{2|R| + |L| + (|L| + |R| - d)} = 1 + \frac{M + d}{3|R| + 2|L| - d}.
        \end{displaymath}
    We can see that the right hand side is an
    increasing function of $d$. Thus, we may again maximize $d$ to get an upper bound.
    \begin{equation}\label{BothCasesNearFirstBound}
        \frac{|NEARFIRST|}{|OPT|} \le 1 + \frac{M + |L| + |R|}{2|R| + |L|}.
    \end{equation}
    In both cases, the bound of Equation \eqref{BothCasesNearFirstBound}
    is valid. Noting that $|L_P| \le |R_P|$ (because $NEARFIRST$ chose to go 
    to the left first), we now use Lemma \ref{lem:LOCATIONS} to finalize 
    our proof. 
        \begin{displaymath}
            |L_P| \le |R_P| \implies |R| \ge |L| - 2M \implies 2|R| \ge |L| + |R| - 2M =
        \end{displaymath}
        \begin{displaymath}
            (|L| + |R|)(1 - 2\eta) \implies |R| \ge \frac{1 - 2\eta}{2}(|L| + |R|) \implies
        \end{displaymath}
        \begin{displaymath}
            1 + \frac{M + |L| + |R|}{2|R| + |L|} \le 1 + \frac{(1 + \eta)(|L| + |R|)}{(1 + \frac{1 - 2\eta}{2})(|L| + |R|)} = 1 + \frac{2(1 + \eta)}{3 - 2\eta} \implies
        \end{displaymath}
        \begin{displaymath}
            \frac{|NEARFIRST|}{|OPT|} \le 1 + \frac{2(1 + \eta)}{3 - 2\eta}.
        \end{displaymath}
\end{proof}
We now give the proof of Theorem \ref{the:NearFirstBound}.
\begin{proof}[Proof of Theorem \ref{the:NearFirstBound}.]
By Lemma $\ref{lem:OpenRobustness}$, $NEARFIRST$ is 3-robust.
Additionally, by Lemma \ref{lem:NearFirstSmoothness}, $NEARFIRST$
is $(1 + \frac{2(1 + \eta)}{3 - 2\eta})$-smooth. Thus, Theorem \ref{the:NearFirstBound} holds.
\end{proof}

%% file: Proofs/OpenAttackLOCATIONS.tex
The logic behind the attack we give here is exactly the 
same as the one used to prove Theorem \ref{the:CompetitiveRatio}.
There are two main differences demanded by the nature of the open 
variant. 

One is that the first phase is shorter in this attack. 
Instead of stopping when $ALG$ exits $[L_U(t), R_U(t)]$, the phase 
now stops when $ALG$ exits the interval $[3L_U(t) + 2, 3R_U(t) - 2]$.
This is so that both options of $ALG$ (switch to the other side or wait for the "frozen" requests) are equally hurtful.

The other difference lies in the release times of the 
second phase. Each request on the side chosen by $ALG$ 
now has its release time delayed to $2 + d$ (instead of $4 - d$), where $d$ is 
the request's distance from the origin. This is so $OPT$ can 
finish by $t = 3$, which is the fastest possible even if all 
requests are released immediately.

In the following, we assume without loss of generality that $ALG$ exits the 
interval $[3L_U(t) + 2, 3R_U(t) - 2]$ from the left side. We define the family $F_O$ of inputs just like we defined $F_C$ in the proof of Theorem \ref{the:CompetitiveRatio}. Of course, the release times are different as explained in the previous paragraphs.
 
An immediate observation that we have already mentioned is the following lemma.

\begin{lemma}\label{lem:OpenAttackOPT3}
    For any instance $f$ in the family $F_O$, $|OPT| = 3$.
\end{lemma}
\begin{proof}
    Since $ALG$ exits the interval from the left side, each request $q_r$ on the right side is released at time
    $2 - |q_r|$. Thus, by moving to $1$ and back, $OPT$ serves all the requests on the right side. Additionally, each request $q_l$ on the left side is released no later than $2 + |q_l|$, allowing $OPT$ to serve these requests by just moving to $-1$ after reaching the origin at $t = 2$. Therefore, $OPT$ can serve all the requests by time $3$, i.e. $|OPT| = 3$.
\end{proof}

The next piece of the puzzle is a lower bound 
on $ALG$. This is given by the following lemma.

\begin{lemma}\label{lem:ALGSlowOpen}
    For any instance $f$ in the family $F_O$ and any algorithm $ALG$,
    $|ALG| \ge \frac{13}{3} - 3\alpha(f)$,
    where $\alpha(f)$ is the distance between consecutive requests in $f$.
\end{lemma}
To prove this lemma, we will use some claims,
many of which are very similar to claims used for 
Lemma \ref{lem:ALGSlow}. To present these claims, we introduce two important terms. 

We denote with $t_{commit}$ the start time of the second phase, i.e. 
\begin{displaymath}
    t_{commit} = min(\{t \: : \: \lnot (3L_U(t) + 2 < pos_{ALG}(t) < 3R_U(t) - 2)\}).
\end{displaymath}
Additionally, we draw attention to the value $d_{commit} = |L_U(t_{commit})|$, which is very important for the attack.

We now present the claims which are very similar to those used for the closed variant. Claims 
\ref{claim:OpenAttackTCommitInterval}, \ref{claim:OpenAttackPosCommitInterval} and \ref{claim:OpenAttackTDCommit} can be seen in the same way as Claims \ref{claim:commit_1_2}, \ref{claim:posCommitInterval} and \ref{claim:timeCommitInterval}, respectively.

\begin{claim}\label{claim:OpenAttackTCommitInterval}
    $1 \le t_{commit} \le 1 + \frac{1}{3}$.
\end{claim}

\begin{claim}\label{claim:OpenAttackPosCommitInterval}
    $pos(t_{commit}) \le 3L_U(t_{commit}) + 2$.
\end{claim}

\begin{claim}\label{claim:OpenAttackTDCommit}
 For any instance $f$ in the family $F_O$, $t_{commit} \ge 2 - d_{commit} - \alpha(f)$.
\end{claim}

Moreover, we also have the following claim.
\begin{claim}\label{claim:OpenAttackDCommit23}
    For any instance $f$ of the family $F_O$, $d_{commit} \ge \frac{2}{3} - \alpha(f)$.
\end{claim}
\begin{proof}
    By Claim \ref{claim:OpenAttackTDCommit}, we have $d_{commit} \ge 2 - t_{commit} - \alpha(f)$.
    Additionally, Claim \ref{claim:OpenAttackTCommitInterval} implies that $t_{commit} \le \frac{4}{3}$.
    The claim follows.
\end{proof}

We are now ready to prove Lemma \ref{lem:ALGSlowOpen}.
\begin{proof}[Proof of Lemma \ref{lem:ALGSlowOpen}.]
    We distinguish cases based on the order in which 
    $ALG$ chooses to serve $-1, 1, L_U(t_{commit})$ after $t_{commit}$.
    \textbf{Case 1.} $ALG$ serves $1$ before $-1$. By Claims \ref{claim:OpenAttackTDCommit}, \ref{claim:OpenAttackPosCommitInterval} and \ref{claim:OpenAttackDCommit23}, this takes at least
    \begin{displaymath}
        |ALG| \ge t_{commit} + |pos(t_{commit})| + 2 + 1 \ge
    \end{displaymath}
    \begin{displaymath}
        2 - d_{commit} - \alpha(f) + |3L_U(t_{commit}) + 2| + 2 + 1 = 
    \end{displaymath}
    \begin{displaymath}
        2 - d_{commit} -\alpha(f) + 3d_{commit} - 2 + 2 + 1 = 3 + 2d_{commit} -\alpha(f)\ge
    \end{displaymath}
    \begin{displaymath}
        3 + \frac{4}{3} - 3\alpha(f) \ge \frac{13}{3} - 3\alpha(f).
    \end{displaymath}
    \textbf{Case 2.} $ALG$ serves $L_U(t_{commit})$ before $1$. By the definition 
    of the second phase's release times and Claim \ref{claim:OpenAttackDCommit23}, we have
    \begin{displaymath}
        |ALG| \ge rel(L_U(t_{commit})) + |L_U(t_{commit})| + 1 = 2 + d_{commit} + d_{commit} + 1
    \end{displaymath}
    \begin{displaymath}
        3 + 2d_{commit} \ge \frac{13}{3} - 2\alpha(f).
    \end{displaymath}
    \textbf{Case 3.} $ALG$ serves in the order $-1$, $1$, $L_U(t_{commit})$. By Claim \ref{claim:OpenAttackDCommit23},
    we easily obtain
    \begin{displaymath}
        |ALG| \ge 2 + 2 + \frac{2}{3} -\alpha(f) \ge \frac{13}{3} - \alpha(f).
    \end{displaymath}
    These cases are exhaustive and thus Lemma \ref{lem:ALGSlowOpen} follows.
\end{proof}
We can now use Lemmas \ref{lem:OpenAttackOPT3} and \ref{lem:ALGSlowOpen} to prove Theorem \ref{the:CompetitiveRatioOpen}.
\begin{proof}[Proof of Theorem \ref{the:CompetitiveRatioOpen}.]
By Lemma \ref{lem:OpenAttackOPT3}, we have $|OPT| = 3$. On the other hand, by Lemma \ref{lem:ALGSlowOpen}, we see that $|ALG| \ge \frac{13}{3} - 3\alpha(f)$. Thus, we obtain a competitive 
ratio of at least $\frac{\frac{13}{3} - 3\alpha(f)}{3}$. For arbitrarily small $\alpha(f)$,
this value can be arbitrarily close to $1.\overline{44}$, proving the Theorem.
\end{proof}

%% file: Sections/OpenLF.tex
\subsection{The \textit{LOCATIONS+FINAL} prediction model}\label{section:openLF}
In our final setting we consider the open variant under the $LF$
prediction model. We give the $PIVOT$ algorithm, which is $1.\overline{33}$-competitive
with perfect predictions and is also smooth and robust. We also reuse the attack strategy described 
for the closed variant to achieve a lower bound of $1.25$.

\paragraph*{The $PIVOT$ algorithm.}
The final algorithm we present works in the same way as 
$NEARFIRST$, except for the order in which it focuses on the two sides 
of the origin. Instead of heading to the near extreme first, $PIVOT$ prioritizes 
the side whose extreme is further away from the predicted endpoint of $OPT$, which is provided
by the $LF$ prediction model. The pseudocode for $PIVOT$ is given below. Note that 
$P_{f'}$ refers to the element in $P$ with label $f'$.

\begin{nscenter}
\begin{algorithm}
    \caption{$PIVOT$ update function.}\label{alg:Pivot}
    \SetKwInOut{Input}{Input}
        \SetKwInOut{Output}{Output}
        \Input{Current position $pos$, set $O$ of unserved released requests, set $P$ of predictions, label $f'$
        of $OPT$'s predicted endpoint.}
        \Output{A series of (unit speed) moves to carry out until the next request is released.}
    $P' \gets$ the unreleased predictions in $P$\;
    \If{$P'$ is empty}{
        \lIf{$pos < \frac{max(O) + min(O)}{2}$}{
            \Return $move( min(O)) \oplus move(max(O))$
            }
            \lElse{
            \Return $move(max(O)) \oplus move( min(O))$
            }
    }
    \lIf{$P_{f'} > \frac{max(P) + min(P)}{2}$}{
        \Return $move(min(P' \cup O)) \oplus move(min(P'))$
    }
    \lElse{
        \Return $move(max(P' \cup O)) \oplus move(max(P'))$
    }

\end{algorithm}
\end{nscenter}

As for the previous algorithms, we show a theorem that pertains to $PIVOT$'s competitive 
ratio for different values of the $\eta$ and $\delta$ errors.
\PivotBound*
The proof is very similar to the one used for $NEARFIRST$'s competitive ratio. In fact, the robustness is shown 
in exactly the same way. For the consistency/smoothness, the delay is bounded by $M + d$ in the same way, where $d$ is the distance of the last request $q_f$ served by $OPT$ to the extreme served second by $PIVOT$. The same lower bounds for $|OPT|$ hold as well. We additionally bound $d$ as a function of the error-dependent values $\Delta$ and $M$. When there is no error, we can bound $d$ to be at most $\frac{|L| + |R|}{2}$ instead of $|L| + |R|$, which gives a better competitive ratio than that of $NEARFIRST$. An important distinction is that we do not make use of Lemma \ref{lem:LOCATIONS}, since the algorithm does not consider the amplitudes of $L$ and $R$.

The formal proof of Theorem \ref{the:PivotBound} is given here.
\input{Proofs/PIVOT}

\paragraph*{A $1.25$-attack.}
We make use of the original attack strategy of Section \ref{section:closedLocations} yet again to obtain a lower 
bound for this setting. Our final theorem is presented here.
\LFAttack*
\input{Proofs/LFAttack}

%% file: Proofs/PIVOT.tex
The robustness part of Theorem \ref{the:PivotBound} is implied by Lemma \ref{lem:OpenRobustness}.
To prove Theorem \ref{the:PivotBound}, it remains to show the consistency/smoothness part, which is given by the following lemma.

\begin{lemma}\label{lem:PivotSmoothness}
    The algorithm $PIVOT$ is $f(\eta, \delta)$-smooth for $3 - 2(\delta + 2\eta) > 0$, where
    \begin{displaymath}
        f(\eta, \delta) = 1 + \frac{1 + 2(\delta + 3\eta)}{3 - 2(\delta + 2\eta)}.
    \end{displaymath}
\end{lemma}
    We assume without loss of generality that $PIVOT$ first serves the left extreme.  
    By Lemma \ref{lem:OpenDelay}, we see that $|PIVOT| - |OPT| \le M + d$,
    where $d$ is the distance of $R$ to the request 
    $q_f$ served last by $OPT$. We now show a bound on the value 
    of $d$ that depends on $\eta$ and $\delta$.
    \begin{claim}\label{claim:dBound}
        $d = |R - q_f| \le (|L| + |R|)(\frac{1}{2} + \delta + 2\eta)$.
    \end{claim}
    \begin{proof}
    We first show two inequalities that will be used later to prove the claim. Note that $|R - R_P| \le M$ and $|L - L_P| \le M$ by Claim \ref{claim:extremesM}. The first inequality is
    \begin{displaymath}
        |R - R_P| \le M \implies R - R_p \le M \implies R_P - q_f \ge R - q_f - M \implies
    \end{displaymath}
    \begin{equation}\label{R_to_q_f}
        R_P - q_f \ge |R - q_f| - M.
    \end{equation}
    Similarly, we see that
    \begin{displaymath}
        |L - L_P| \le M \implies L_P - L \ge -M \implies q_f - L + M \ge q_f - L_P
    \end{displaymath}
    \begin{equation}\label{L_to_q_f}
        |q_f - L| + M \ge q_f - L_P.
    \end{equation}
    Because of $PIVOT$'s choice to go left first, we see that
    \begin{displaymath}
        |R_P - p_{f'}| \le |L_P - p_{f'}| \implies R_P - p_{f'} \le p_{f'} - L_P \implies p_{f'} \ge \frac{R_P + L_P}{2} \implies 
    \end{displaymath}
    \begin{displaymath}
       q_{f'} \ge \frac{R_P + L_P}{2} - M \implies q_f \ge \frac{R_P + L_P}{2} - M -\Delta \implies
    \end{displaymath}
    \begin{displaymath}
        q_f - L_P + M + \Delta \ge R_P - q_f - M - \Delta \overset{\eqref{R_to_q_f}, \eqref{L_to_q_f}}{\implies}
    \end{displaymath}
    \begin{displaymath}
        |q_f - L| + 2M + \Delta \ge |R - q_f| - 2M - \Delta \implies 
    \end{displaymath}
    \begin{displaymath}
        |q_f - L| + |R - q_f| + 2M + \Delta \ge 2|R - q_f| - 2M - \Delta \implies
    \end{displaymath}
    \begin{displaymath}
        d = |R - q_f| \le (|L| + |R|)(\frac{1}{2} + \delta + 2\eta).
    \end{displaymath}
    \end{proof}
    Using Claim \ref{claim:dBound}, we can prove Lemma \ref{lem:PivotSmoothness}.
    \begin{proof}[Proof of Lemma \ref{lem:PivotSmoothness}.]
    We distinguish two cases based on the order in which $OPT$
    serves the requests of the set $\{L, R, q_f\}$.
    \textbf{Case 1.}
    $OPT$ serves in the order $L, R, q_f$. Thus, we know that
    \begin{displaymath}
        \frac{|PIVOT|}{|OPT|} \le 1 + \frac{M + d}{2|L| + |R| + d}.
    \end{displaymath}
    The derivative of the right part
    with respect to $d$ is
    \begin{displaymath}
        \dfrac{|R| + 2|L| - M}{\left(d + |R| + 2|L|\right)^2}.
    \end{displaymath}
    Since we have assumed $3 - 2(\delta + 2\eta) > 0 \implies \eta < 1$, this value is always positive. Therefore,
    we can set $d$ to the maximum value described in Claim \ref{claim:dBound} to obtain the following bound.
    \begin{displaymath}
        \frac{|PIVOT|}{|OPT|} \le 1 + \frac{(|L| + |R|)(\frac{1}{2} + \delta + 2\eta) + \eta(|L| + |R|)}{2|L| + |R| + (|L| + |R|)(\frac{1}{2} + \delta + 2\eta)} \le
    \end{displaymath}
    \begin{displaymath}
        1 + \frac{1 + 2\delta + 6\eta}{3 + 2\delta + 4\eta} \le 1 + \frac{1 + 2(\delta + 3\eta)}{3 - 2(\delta + 2\eta)}.
    \end{displaymath}
    \textbf{Case 2.} $OPT$ serves in the order $R, L, q_f$. In that case, we have
    \begin{displaymath}
        \frac{|PIVOT|}{|OPT|} \le 1 + \frac{d + M}{2|R| + |L| + (|L| + |R| - d)} \overset{\ref{claim:dBound}}{\le}
    \end{displaymath}
    \begin{displaymath}
        1 + \frac{(|L| + |R|)(\frac{1}{2} + \delta + 2\eta) + \eta(|L| + |R|)}{2|R| + |L| + (|L| + |R|)(\frac{1}{2} - \delta - 2\eta)} \le 1 + \frac{1 + 2\delta + 6\eta}{3 - 2\delta - 4\eta} =
    \end{displaymath}
    \begin{displaymath}
        1 + \frac{1 + 2(\delta + 3\eta)}{3 - 2(\delta + 2\eta)}.
    \end{displaymath}
    In both cases, we have shown the smoothness bound. Therefore, the proof is complete.
\end{proof}

We now give the proof of Theorem \ref{the:PivotBound}. 
\begin{proof}[Proof of Theorem \ref{the:PivotBound}.]
By Lemma \ref{lem:OpenRobustness},
$PIVOT$ is 3-robust. Also, by Lemma \ref{lem:PivotSmoothness}, it is $(1 + \frac{1 + 2(\delta + 3\eta)}{3 - 2(\delta + 2\eta)})$-smooth. Thus, Theorem \ref{the:PivotBound} follows.
\end{proof}

%% file: Proofs/LFAttack.tex
To prove this theorem, we will again utilize 
the attack strategy given in the proof of Theorem \ref{the:CompetitiveRatio}. The inputs generated are the same, except for a new request $q_0$ placed at the origin and released at $t = 4$. Let ${F_C}'$ denote this new family of inputs. We observe the following lemmas.
\begin{lemma}\label{lem:LFOPT4}
    For any instance $f$ in the family ${F_C}'$,
    $|OPT| = 4$.
\end{lemma}
\begin{proof}
We see that the requests of the side which 
    $ALG$ did not exit from are released such that $OPT$
    can serve them all and return to the origin by $t = 2$.
    Additionally, the other side's requests are released 
    such that $OPT$ never has to stop for them either, i.e. 
    it can serve them all and return to the origin by $t = 4$. The request on the origin is released at exactly $t = 4$, so this is also served right as $OPT$ returns to the origin from the second trip. Thus, $|OPT| = 4$.
\end{proof}
In the following, $\alpha(f)$ will refer to the distance of consecutive requests in $f$, disregarding $q_0$.
\begin{lemma}\label{ALGSlowLF}
    For any instance $f$ in the family ${F_C}'$,
    $|ALG| \ge 5 - 2\alpha(f)$.
\end{lemma}
\begin{proof}
Suppose for the sake of contradiction that $|ALG| < 5 - 2\alpha(f)$. We can see that $|pos_{ALG}(|ALG|)| \le 1$. Thus, an algorithm $ALG'$ could copy $ALG$ until it serves all requests and then return to the origin. That would mean that $ALG'$ solves the closed variant of $f$ such that $|ALG'| < 6 - 2\alpha(f)$. Observe that there exists an instance $f' \in F_C$ that is identical to $f$ except for $q_0$. We can see that $ALG'$ also solves $f'$ in less than $6 - 2\alpha(f) = 6 - 2\alpha(f')$ time units, since $f'$ only contains a subset of the requests in $f$. 
Therefore, we have a contradiction to Lemma \ref{lem:ALGSlow}.
\end{proof}
We can now prove Theorem \ref{the:LFAttack}.
\begin{proof}[Proof of Theorem \ref{the:LFAttack}.]
By Lemma \ref{lem:LFOPT4}, we see that $|OPT| = 4$. We also see that $|ALG| \ge 5 - 2\alpha(f)$ by Lemma \ref{ALGSlowLF}. Thus, we get a competitive ratio of at least $\frac{5 - 2\alpha(f)}{4}$, which can be arbitrarily close to $1.25$, concluding the proof.
\end{proof}

%% file: Experiments.tex
We have generated synthetic instances and corresponding predictions and tested our algorithms on them. In this
section, we explain how this data was generated and present the results we acquired.

Note that mirroring of the positions of the requests and/or uniform scaling of the positions and release times does not affect the competitive ratio of any algorithm. Therefore, we choose to generate the inputs as explained below.
\paragraph{Generating inputs.} 
In the following, any reference of randomness will correspond to a uniform distribution. We have a maximum number of requests $n_{max} \ge 2$ and a maximum release time $r_{max}$.
Our generator first randomly chooses an integer number of requests $n \in [2, n_{max}]$. 
Then it randomly chooses a value $c' \in [1, c]$. We then generate the \textit{positions}
of the requests as follows. We always have a request at $-1$ and one at $c'$. The other 
$n - 2$ requests are randomly placed in the interval $[-1, c']$. The release time 
of each request is randomly chosen from $[0, r_{max}]$.

\paragraph{Generating predictions.} 
We now briefly explain how the predictions of the $LOCATIONS$ model 
are generated. Each input generated also comes 
with a prediction "mould". This mould contains $n$ scalars $m_i \in [-1, 1]$, one 
for each request. At least one of these scalars has an absolute value of $1$.
For a given error $\eta$, we calculate $M$ and then add an offset of $m_i \cdot M$ to the position 
of request $q_i$ to get the prediction $p_i$. In this way, at least one prediction 
is guaranteed to have a distance of $M$ to its associated request.

For the $LF$ prediction model, we simply try each label of the generated input
as a different prediction. 
Each label choice corresponds to 
a different error $\delta$, which 
is calculated after choosing the label.

\paragraph{Results.}
We generated $7500$ random input-predictions pairs with at most $20$ requests. A value 
of $c = 2$ was chosen, since higher values of $c$ in general only benefit our algorithms. The maximum release time was set to $6$. Again, higher release times in general lead to better competitive ratios for our algorithms, because they increase $|OPT|$.

The error $\eta$ of these predictions varied from $0$ to $1$. We ran $FARFIRST$ and $NEARFIRST$ on each 
of these instances. Additionally, for each of these instances, we ran $PIVOT$ with each of the instance's request labels 
as the prediction of the $LF$ prediction model. Thus, the $PIVOT$ algorithm was ran approximately $75000$ times.

We did not compare the results of our algorithms to the classical algorithms because that would 
be unfair. That is because our algorithms have the benefit of knowing the number of requests $n$
which helps in practice, even if the theoretical lower bounds are almost identical. In contrast, 
the theoretically optimal classical online algorithms resort to 
waiting techniques, which in turn almost always maximizes their competitive ratio to the theoretical bound.

The experiments were executed on a typical modern laptop computer (CPU: AMD Ryzen 7 4700U 2.0 Ghz 8 cores, RAM: 16GB). The execution time did not exceed $2$ minutes.
We present our results via various graphs in the following subsections.

\subsection{\textit{FARFIRST}}\label{section:farFirstExperiments}
Figure \ref{figure:farFirstSimple} shows the maximum competitive ratio observed for the $FARFIRST$ algorithm
in all instances with error $\eta$ up to the value of the $x$ axis. In figure \ref{figure:farFirstPercentile}, we have also provided a plot that depicts the maximum competitive ratio observed for $x$ \% of the best instances with error $\eta$ up to the value of the $y$ axis. We note that the grid turns red near the very edge, which means that high competitive ratios are rare.

\begin{figure}[ht]
\centering
\begin{minipage}{.48\textwidth}
  \centering
  \includegraphics[width=1\linewidth]{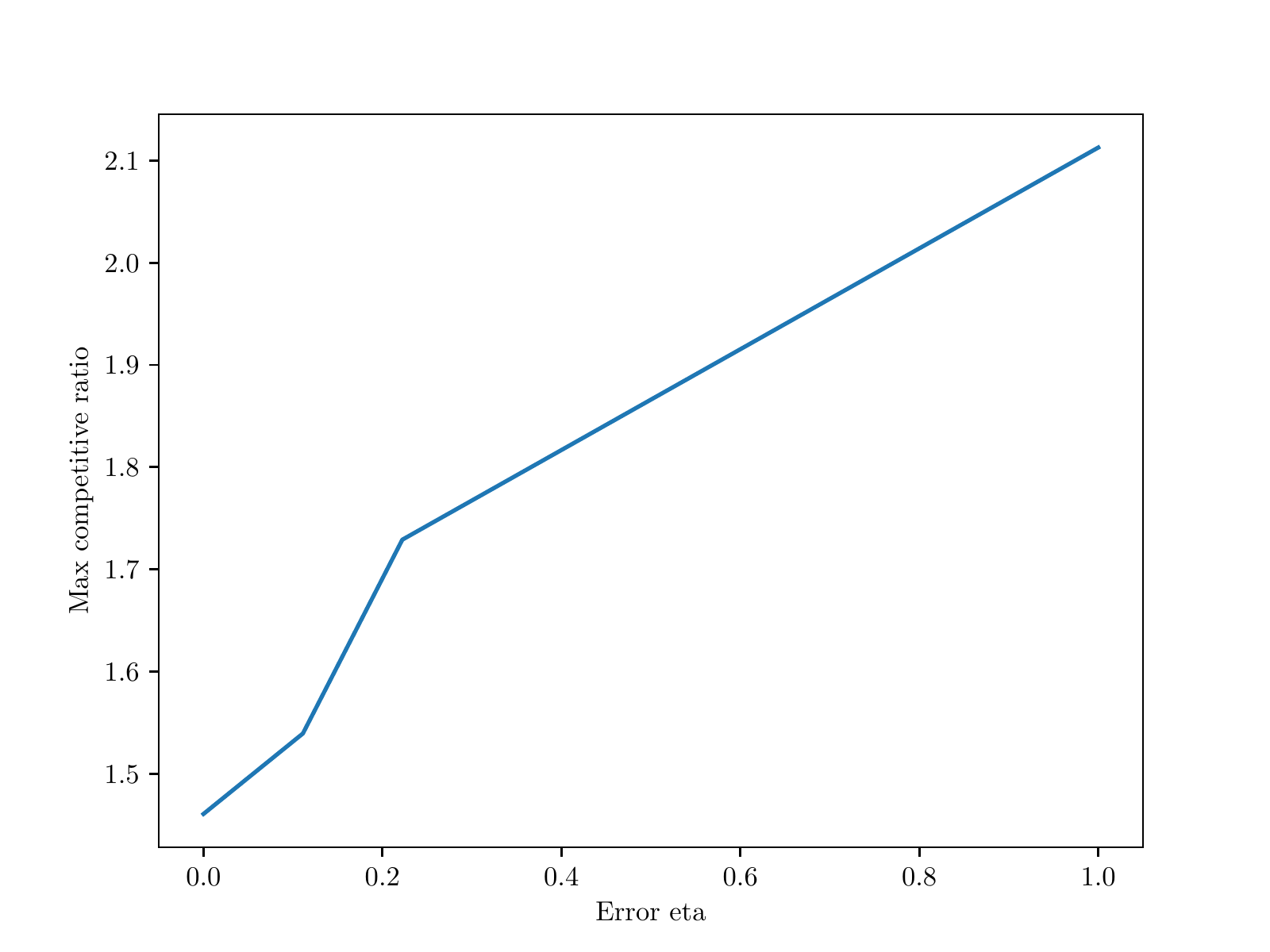}
  \captionof{figure}{\textit{FARFIRST}'s competitive ratio for increasing error. As can be seen in the figure, the competitive ratio never surpasses $\approx 2.15$ for $\eta \le 1$. Additionally, we find that the competitive ratio even with zero error is close to the theoretical upper bound of $1.5$. It should also be noted that the theoretical lower bound of $1.64$ (without predictions) is broken for $\eta$ roughly up to $0.2$.}
  \label{figure:farFirstSimple}
\end{minipage}%
\hfill
\begin{minipage}{.48\textwidth}
  \centering
  \includegraphics[width=1\linewidth]{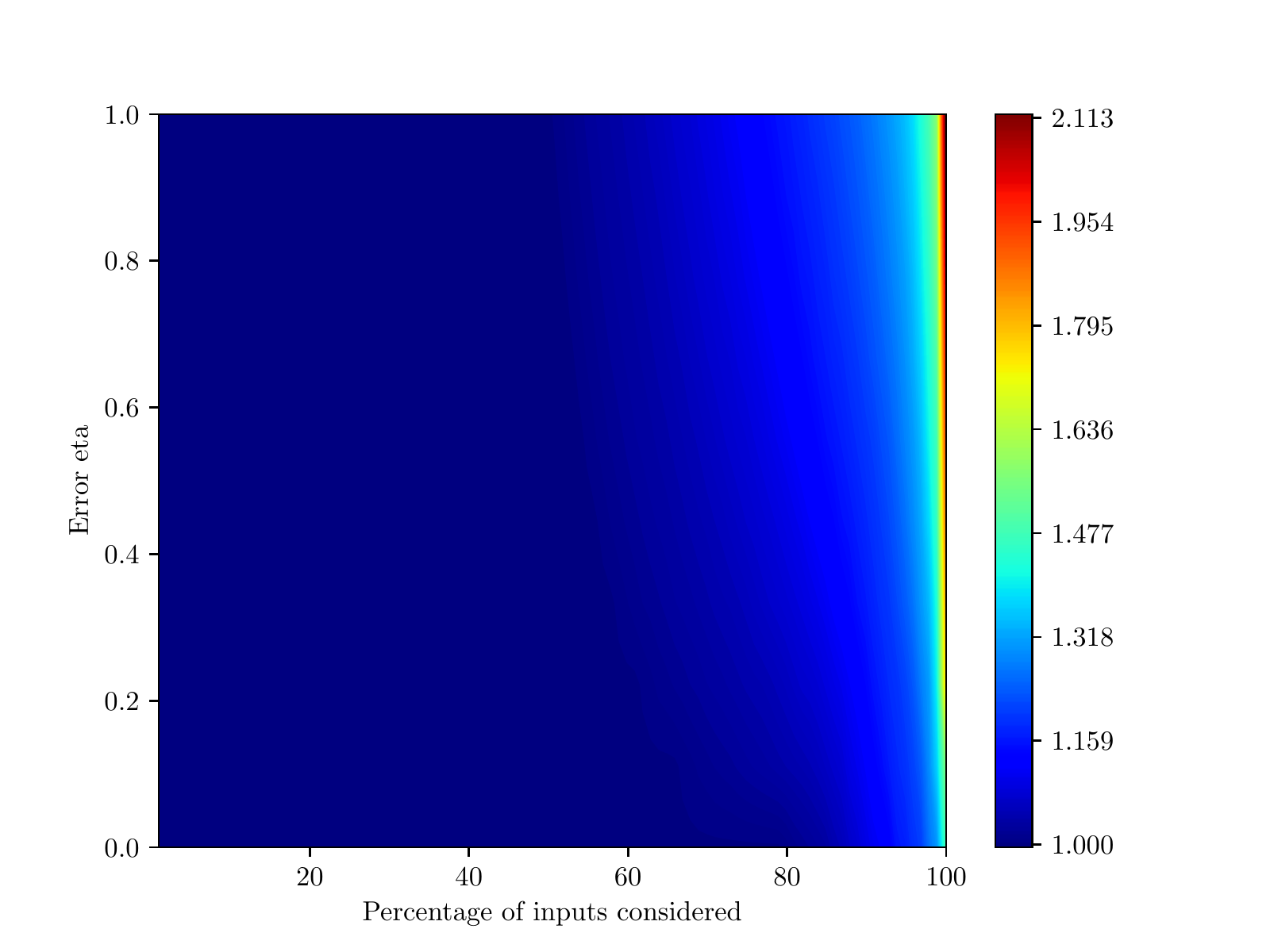}
  \captionof{figure}{\textit{FARFIRST}'s competitive ratio for increasing error and 
  percentage of inputs considered, sorted by the competitive ratio 
  that $FARFIRST$ obtains on them. The narrowness of the red portion of the grid suggests
  that high competitive ratios are rare. We note that the colors 
  of the grid are generally blue, i.e.
  $FARFIRST$ exhibits a relatively low competitive ratio in most cases.}
  \label{figure:farFirstPercentile}
\end{minipage}
\end{figure}

\subsection{\textit{NEARFIRST}}
The figures presented here are analogous to those of Section \ref{section:farFirstExperiments}. Figure \ref{figure:nearFirstSimple} shows the maximum competitive ratio observed for the $NEARFIRST$ algorithm
in all instances with error $\eta$ up to the value of the $x$ axis. In figure \ref{figure:nearFirstPercentile}, a plot analogous to that seen in figure \ref{figure:farFirstPercentile} is shown for the $NEARFIRST$ algorithm. The red portion of the grid is again quite limited as in the case for $FARFIRST$.

\begin{figure}[ht]
\centering
\begin{minipage}{.48\textwidth}
  \centering
  \includegraphics[width=1\linewidth]{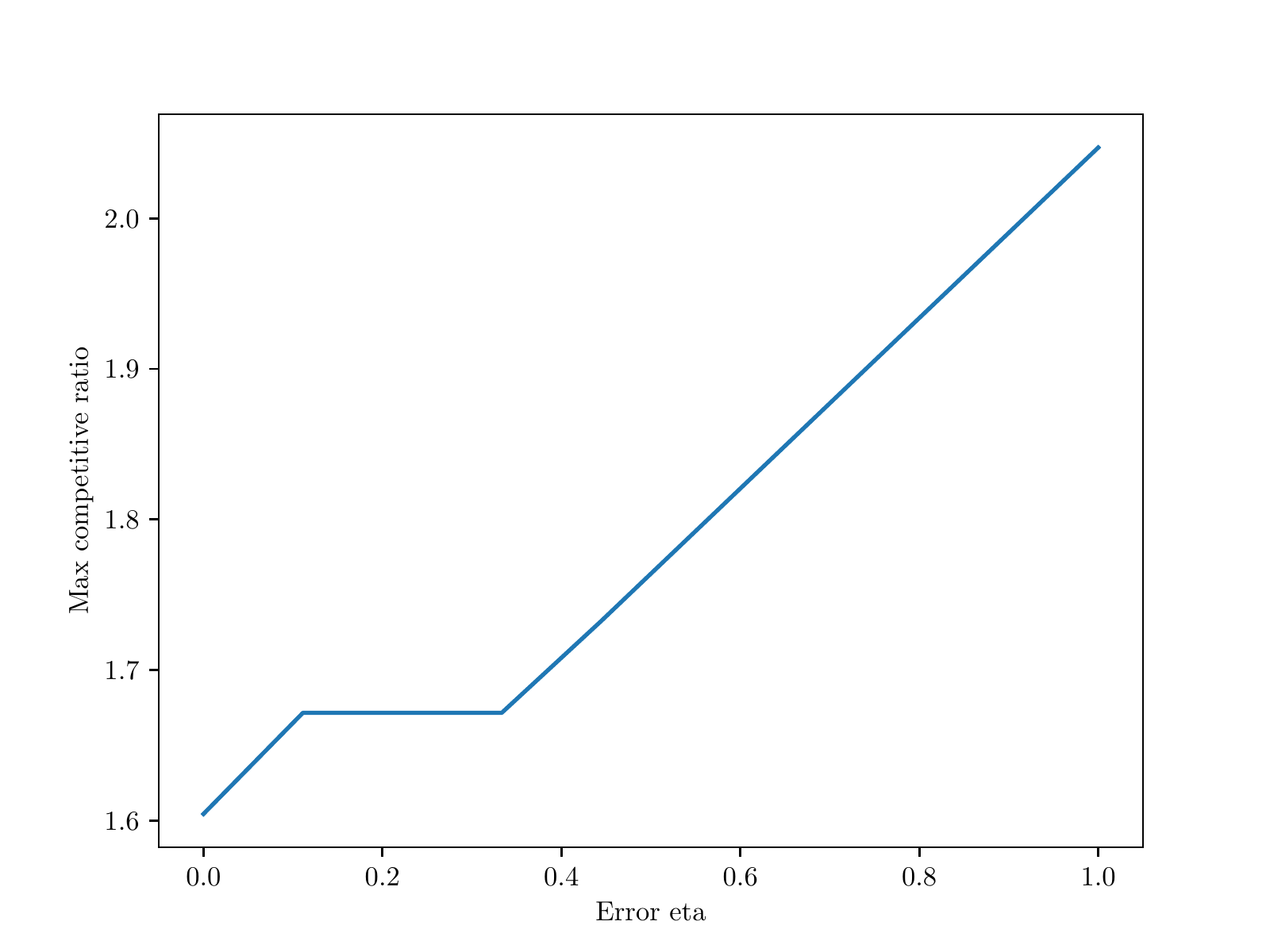}
  \captionof{figure}{\textit{NEARFIRST}'s competitive ratio for increasing error. As can be seen in the figure, the competitive ratio never surpasses $\approx 2.05$ for $\eta \le 1$. Additionally, we find that the competitive ratio even with zero error is close to the theoretical upper bound of $1.\overline{66}$. It should also be noted that the theoretical lower bound of $2$ (without predictions) is broken even for $\eta$ very close to $1$.}
  \label{figure:nearFirstSimple}
\end{minipage}%
\hfill
\begin{minipage}{.48\textwidth}
  \centering
  \includegraphics[width=1\linewidth]{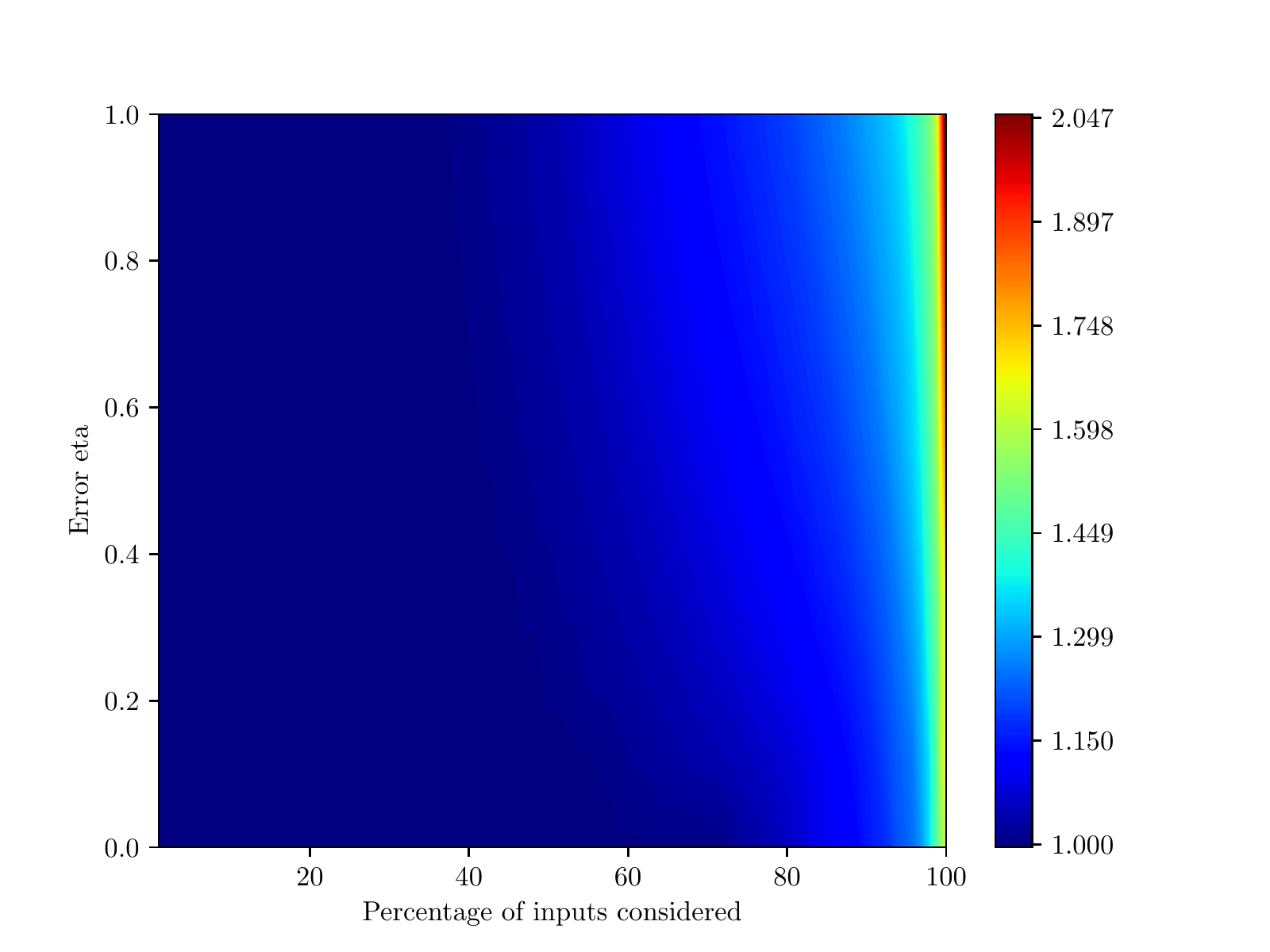}
  \captionof{figure}{\textit{NEARFIRST}'s competitive ratio for increasing error and 
  percentage of inputs considered, sorted by the competitive ratio 
  that NEARFIRST obtains on them. The narrowness of the red portion of the grid suggests
  that high competitive ratios are rare. We note that the colors 
  of the grid are generally blue, i.e.
  $NEARFIRST$ exhibits a relatively low competitive ratio in most cases.}
  \label{figure:nearFirstPercentile}
\end{minipage}
\end{figure}

\subsection{\textit{PIVOT}}
In this final subsection we condsider the $PIVOT$ algorithm. In figure \ref{figure:pivotSimple}, the color of the pixel in coordinates $(x, y)$ corresponds to the maximum competitive ratio observed for all instances with errors $\delta \le x$ and $\eta \le y$. We should explain here that the colors change abruptly in this figure since the $\delta$ error does not vary smoothly in the generated predictions. This is because we only have a discrete set of choices for the label $f'$ of the $LF$ prediction model through which $\delta$ is calculated.

\begin{figure}[ht]
\centering
\includegraphics{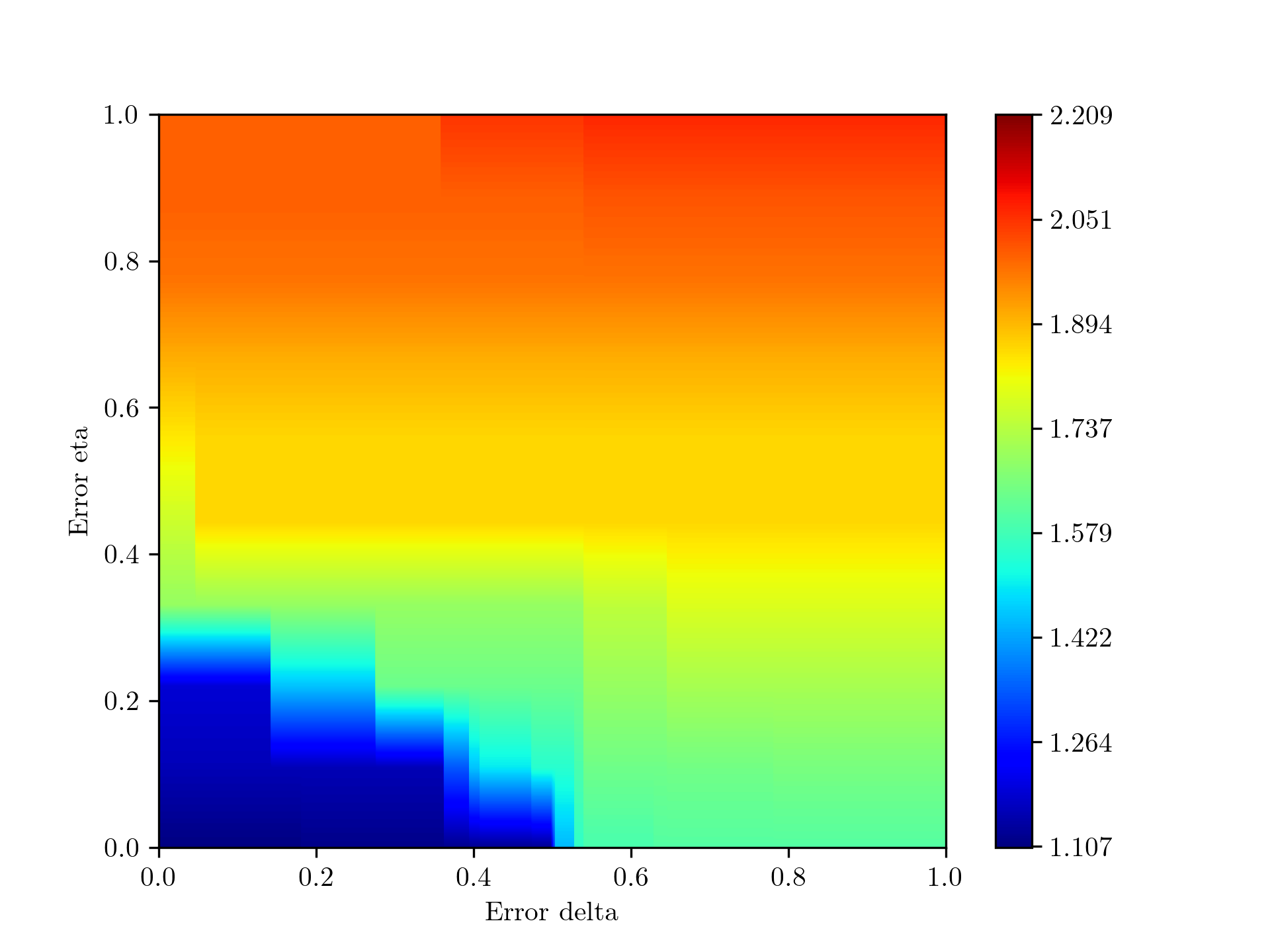}
\caption{\textit{PIVOT}'s competitive ratio for increasing errors $\delta$ and $\eta$. The color of the pixel in coordinates $(x, y)$ corresponds to the maximum competitive ratio observed for all instances with errors $\delta \le x$ and $\eta \le y$. We observe that the competitive ratio is more sensitive to $\eta$ than to $\delta$, as was to be expected by the corresponding theoretical bound. With perfect predictions, the maximum competitive ratio is not greater than $\approx 1.11$, which is considerably lower 
than the theoretical upper bound of $1.\overline{33}$. In general, the competitive ratio increases smoothly along the main diagonal of the grid. Finally, $PIVOT$'s competitive ratio surpasses the lower bound of $2$ (without predictions) only for large values of $\delta$, $\eta$.}\label{figure:pivotSimple}
\end{figure}

%% file: Conclusions.tex
We have examined the online TSP on the line and provided lower bounds as well as algorithms for three different learning-augmented settings.
An immediate extension of our results would be to bridge the gap between the lower and upper bounds we have shown for the open variant. Also, it would be interesting to establish error-dependent lower bounds and/or optimal consistency-robustness tradeoffs. Moreover, an improvement would be to remove the assumption of knowing the number of requests $n$. A technique that could perhaps allow an algorithm to achieve that is to periodically make sure that the algorithm terminates in case no new requests appear. Finally, more general versions of online TSP could be investigated like the case of trees.

%% file: Appendix.tex
\section{Lemmas for the case of known \textit{n}}

\begin{lemma}
For any $\epsilon > 0$, no algorithm can be $(2-\epsilon)$-competitive
for open online TSP on the line without predictions when the 
number of requests $n$ is known. Also, 
there exists an algorithm that matches this lower bound.
\end{lemma}
\begin{proof}
A very simple attack can be 
used to show the lower bound of $2$. If $pos_{ALG}(1) \le 0$, we present a request 
at $1$ with a release time of $1$. In the other case, the request's position is $-1$.
It is easy to see that $|OPT| = 1$, while $|ALG| \ge 2$, proving the bound. There also 
exists a very simple algorithm that matches this bound. Such an algorithm need only wait 
until all requests have been released and then copy $OPT$'s actions, which are at this point 
computable. The waiting part does not take more than $|OPT|$ and neither does the moving 
part, which implies a competitive ratio of $2$ for such an algorithm.
\end{proof}

\begin{lemma}
For any $\epsilon > 0$, no algorithm can be $(1.64-\epsilon)$-competitive
for closed online TSP on the line without predictions when the 
number of requests $n$ is known. Also, 
there exists an algorithm that matches this lower bound.
\end{lemma}
\begin{proof}
By taking 
a close look at the attack strategy described in Section 3.3 of \cite{Ausiello:1.75}, we observe that the 
number of requests is never higher than a specific value $n_{max}$. In fact, it turns out 
that $n_{max} = 3$, i.e. the attack never uses more than 3 requests. We modify this attack 
so that it can also be used when $n$ is known. Any instance of the modified attack will have 
\textit{exactly} $n_{max}$ requests. Thus, the algorithm will be informed that there will indeed be $n_{max}$ requests.

We first give a brief description of the original attack strategy for context. 
Let $\rho = \frac{9 + \sqrt{17}}{8} \approx 1.64$, $I = [-(2\rho - 3), (2\rho - 3)]$, $I' = [-(7 - 4\rho), (7 - 4\rho)]$. Note that $I$ is contained in $I'$ and both of them are contained in $[-1, 1]$. If $pos_{ALG}(1) \notin I$, then a single request at $-1$ or $1$ (released at $t = 1$) suffices to achieve the competitive ratio. Assuming that $pos_{ALG}(1) \in I$, we simultaneously present two requests at $-1$ and $1$ at $t = 1$. At $t = 3$, $ALG$ cannot have possibly served both of these requests. If $pos_{ALG}(3) \in I'$, then another request at $-1$ or $1$ (released at $t = 3$) is sufficient. Therefore, we continue assuming that $pos_{ALG}(3) \notin I'$. This means that $ALG$ is close to one extreme and still has not served the other. When $ALG$ crosses the origin to serve the other extreme at time $3 + x$, a request is placed at either $1 + x$ or $-(1 + x)$ (depending on which extreme $ALG$ has not served). The competitive ratio turns out to be at least $\rho$ in this (final) case also. 

We now describe our modification of this strategy. Initially, the original attack strategy is followed. Let $q_{win}$ be the last request released by 
the original attack strategy, after the release of which the competitive ratio is guaranteed to be at least $1.64$ in 
case no new requests appear. Let $n_{original}$ be the number of requests released via the original attack strategy. If $n_{original} < n_{max}$, then
$n_{max} - n_{original}$ extra requests are released at time $rel(q_{win})$, placed arbitrarily between the origin and $q_{win}$.
These extra requests are served by $OPT$ on the way back from $q_{win}$, without incurring extra cost. In other words, $|OPT|$ does not increase with the 
addition of these requests. Also, $|ALG|$ certainly cannot decrease since we only \textit{added} requests. Therefore, the same 
lower bound holds even for known $n$.

The algorithm is exactly the same as the one for unknown number of requests, since it 
can just ignore the number $n$ and still achieve the same competitive ratio.
\end{proof}

\section{Omitted proofs from section \ref{section:Open}}
In this subsection, we give the formal proofs of two lemmas which we used 
to prove Theorems \ref{the:NearFirstBound} and \ref{the:PivotBound}.
\input{Proofs/OpenLemmata}

%% file: Proofs/OpenLemmata.tex
\OpenRobustness*
    \begin{proof}
    Let $t_{f}$ denote the latest release time for a fixed instance of the problem.
    We assume w.l.o.g. that $pos_{ALG}(t_{f}) \le \frac{L(t_f) + R(t_f)}{2}$. Note that after $t_f$, $ALG$
    will move to $L(t_f)$ and then to $R(t_f)$. Thus, we observe that
    \begin{equation}\label{OpenCleanupTime}
        |ALG| = t_f + |pos(t_f) - L(t_f)| + |L(t_f) - R(t_f)|.
    \end{equation}
    We distinguish two 
    cases based on the position of $ALG$ at time $t_f$.
    \textbf{Case 1.} $pos(t_{f}) \ge L(t_{f})$. In this case, we see that
        \begin{displaymath}
    {\eqref{OpenCleanupTime} \implies
    |ALG| = t_f + pos(t_f) - L(t_f) + R(t_f) - L(t_f) \le}
    \end{displaymath}
    \begin{displaymath}
    {t_f + \frac{L(t_f) + R(t_f)}{2} - L(t_f) + R(t_f) - L(t_f) =}
    \end{displaymath}
    \begin{displaymath}
        {t_f + \frac{3(|L(t_f)| + |R(t_f)|)}{2} \le 2.5|OPT| \le 3|OPT|.}
    \end{displaymath}
    \textbf{Case 2.}  $pos(t_{f}) < L(t_{f})$. Similarly, we have
        \begin{displaymath}
    {\eqref{OpenCleanupTime} \implies
    |ALG| = t_f + L(t_f) - pos(t_{f}) + R(t_f) - L(t_f)\le}
    \end{displaymath}
    \begin{displaymath}
        {2t_f + |R(t_f)| \le 3|OPT|.}
    \end{displaymath}
    \end{proof}
\OpenDelay*
    To prove this lemma, we first give some definitions. Note first that $L_U[t]$ is sort of a "checkpoint" for $OPT$, meaning that $OPT$ must be located 
    at $L_U[t]$ for some point in time on or after $t$ in order
    to serve that request. Then, it must move from $L_U[t]$ to $q_f$. This idea helps us keep track of $|OPT|$ so we can 
    compare it with $|ALG|$.

    $D(t)$ denotes the least amount of time 
    necessary to serve all requests to the left of $L_U[t]$ (assuming they have been released) and then move to $L_U[t]$,
    starting at position $pos_{ALG}(t)$. This amounts 
    to
        \begin{displaymath}
            D(t) = |pos_{ALG}(t) - L_O(t)| + |L_O(t) - L_U[t]|.
        \end{displaymath}
This function exhibits a useful bound property. If it 
drops to $M$ or below at some time $t$, it can only increase 
above $M$ again due to a request release. This property is 
described more formally in the following claim. But first, 
another useful definition is given. 

We define $L_P[t]$ as the leftmost prediction that is released on or after $t$. That is,
$L_P[t] =  min(\{p \in P \: : \: rel(\pi(q)) \ge t\})$. If this set is empty, then $L_P[t] = R$.

Using this definition, the following claim can be seen in the same way as Claim \ref{claim:chainSaw}.
\begin{claim}\label{claim:OpenChainSaw}
    Let $t_{drop}$ be a time point such that $D(t_{drop}) \le M$. 
    If $t_{next}$ is the earliest release time of a request 
    after $t_{drop}$, then
        \begin{displaymath}
            D(t') \le M, \: \forall \: t' \in [t_{drop}, t_{next}].
        \end{displaymath}
\end{claim}

We now draw our attention to a point in time that is very 
central to our proof.

Let $t_{release}$ be the latest release time of a request. Note 
    that $L_U[t] = R, \: \forall \: t > t_{release}$. Then, we define
        \begin{displaymath}
            t_{chase} = min\{t : \: t_s \le t \le t_{release}, \: (D(t') > M, \: \forall \: t < t' \le t_{release})\}.
        \end{displaymath}
In essence, similarly to the definition in the previous section, $t_{chase}$ denotes the time after which $ALG$ gets to finish as soon as possible without waiting for predictions or backtracking for requests. In the following, we assume for simplicity and without loss of generality
that $ALG$ \textit{always} serves the requests left to right,
even if at time $t_{release}$ it is clear that going 
to the right first is faster. It is true that our 
algorithm may indeed make such a decision at time $t_{release}$,
but that is a trivial optimization that does not invalidate 
our proof, since it can only decrease $|ALG|$ and by extension,
the value $|ALG| - |OPT|$. Under this assumption, we proceed 
by showing that after $t_{chase}$, $ALG$ moves to $L_O(t_{chase})$
and then straight to $L_U[t_{chase}]$, serving all
intermediate requests on the way. In fact, it also
keeps moving to the right until it reaches $R$ and finishes. The following claim can be seen in the same way as Claim \ref{claim:FastFinalState}.
\begin{claim}\label{claim:OpenFastFinish}
    Let $t' = t_{chase} + D(t_{chase})$. Then, 
    $pos(t') = L_U[t_{chase}]$. 
    
    Also, $|ALG| = t' + |pos(t') - R|$.
\end{claim}
We now give the proof of Lemma \ref{lem:OpenDelay}.
\begin{proof}[Proof of Lemma \ref{lem:OpenDelay}.]
We distinguish two cases.

\textbf{Case 1.} $t_{chase} = t_s$. This easily implies 
that $|ALG| = 2|L| + |R|$ by Claim \ref{claim:OpenFastFinish}. It 
remains to show that $|OPT| \ge 2|L| + |R| - M - d$. 

If $OPT$
follows the order $L \xrightarrow{} R \xrightarrow{} q_f$, then 
    \begin{displaymath}
        |OPT| \ge 2|L| + |R| + d \ge 2|L| + |R| - M - d.
    \end{displaymath}
On the other 
hand, if $OPT$
follows the order $R \xrightarrow{} L \xrightarrow{} q_f$, then
    \begin{displaymath}
        |OPT| \ge 2|R| + |L| + (|L| + |R| - d) \ge 3|R| + 2|L| - d \ge 2|L| + |R| - M - d.
    \end{displaymath}
\textbf{Case 2.} $t_{chase} > t_s$. It can be 
seen then by Claim \ref{claim:OpenChainSaw} that $D(t_{chase}) \le M$.
It is easy to see 
that $|OPT| \ge t_{chase} + |L_U[t_{chase}] - q_f|$. At 
the same time, by Claim \ref{claim:OpenFastFinish} we see that
    \begin{displaymath}
        |ALG| = t_{chase} + D(t_{chase}) + |L_U[t_{chase}] - R| \le
    \end{displaymath}
    \begin{displaymath}
        t_{chase} + M + |L_U[t_{chase}] - q_f| + |q_f - R| \le |OPT| + M + d.
    \end{displaymath}
\end{proof}